\documentclass{amsart}
\usepackage{mathrsfs,amssymb,setspace}
\usepackage{verbatim, enumerate, color, stmaryrd, bbold, bm}

\theoremstyle{plain}
\newtheorem{theorem}{Theorem}[section]
\newtheorem{lemma}[theorem]{Lemma}
\newtheorem{proposition}[theorem]{Proposition}
\newtheorem{corollary}[theorem]{Corollary}

\theoremstyle{definition}
\newtheorem{notation}[theorem]{Notation}
\newtheorem{example}[theorem]{Example}
\newtheorem{definition}[theorem]{Definition}

\theoremstyle{remark}
\newtheorem{remark}[theorem]{Remark}

\usepackage[all]{xy}

\newcommand{\dref}[1]{Definition \ref{#1}}
\newcommand{\lref}[1]{Lemma \ref{#1}}
\newcommand{\tref}[1]{Theorem \ref{#1}}
\newcommand{\pref}[1]{Proposition \ref{#1}}
\newcommand{\cref}[1]{Corollary \ref{#1}}
\newcommand{\rref}[1]{Remark \ref{#1}}

\newcommand{\eref}[1]{Example \ref{#1}}
\newcommand{\sref}[1]{Section \ref{#1}}

\newcommand{\powerset}{\raisebox{.15\baselineskip}{\Large\ensuremath{\wp}}}

\begin{document}


\title[{\tt if-then-else} over possibly non-halting programs and tests]{Axiomatization of {\tt if-then-else} over possibly non-halting programs and tests}
\author[Gayatri Panicker]{Gayatri Panicker}
\address{Department of Mathematics,  Indian Institute of Technology Guwahati, Guwahati, India}
\email{p.gayatri@iitg.ac.in}
\author[K. V. Krishna]{K. V. Krishna}
\address{Department of Mathematics, Indian Institute of Technology Guwahati, Guwahati, India}
\email{kvk@iitg.ac.in}
\author[Purandar Bhaduri]{Purandar
 Bhaduri}
\address{Department of Computer Science and Engineering, Indian Institute of Technology Guwahati, Guwahati, India}
\email{pbhaduri@iitg.ac.in}


\begin{abstract}
In order to study the axiomatization of the {\tt if-then-else} construct over possibly non-halting programs and tests, this paper introduces the notion of \emph{$C$-sets} by considering the tests from an abstract $C$-algebra.  When the $C$-algebra is an ada, the axiomatization is shown to be complete by obtaining a subdirect representation of $C$-sets. Further, this paper considers the equality test with the {\tt if-then-else} construct and gives a complete axiomatization through the notion of \emph{agreeable $C$-sets}.
\end{abstract}



\maketitle

\section*{Introduction}

Being one of the fundamental constructs, the conditional expression {\tt if-then \sloppy -else} has received considerable importance in programming languages. It plays a vital role in the study of program semantics. One of the seminal works in the axiomatization of this conditional expression was by McCarthy in \cite{mccarthy63}, where he gave an axiom schema for the determination of the semantic equivalence between any two conditional expressions. Since then several authors have studied the axiomatization of {\tt if-then-else} in different contexts.

Following McCarthy's approach, Igarashi in \cite{igarashi71} studied a formal system comprising ALGOL-like statements including various programming features along with {\tt if-then-else} with predicates. The two systems were shown to be equipollent by de Bakker in \cite{debakker69}, i.e., axioms of one could be derived from the other. In \cite{sethi78} Sethi gave a different framework to determine the semantic equivalence of statements of the form {\tt if $E = F$ then $G$ else $H$}. In \cite{pigozzi91} Pigozzi gave an axiomatization of the theory of the equality test algebras appended with {\tt if-then-else}, where the test is purely $T$ ({\tt true}) or $F$ ({\tt false}). He gave a finite axiom scheme for the quasi-equational theory of equality test algebras and another finite axiom scheme for the equational theory of {\tt if-then-else} algebras. In \cite{bergman91} Bergman studied the sheaf-theoretic representation of sets equipped with an action of a Boolean algebra. This Boolean action was in fact the {\tt if-then-else} function. This approach was adopted in \cite{stokes98} by Stokes who obtained a representation theorem for the Boolean algebra case of {\tt if-then-else} algebras of \cite{manes93}.

In \cite{jackson09} Jackson and Stokes gave a complete axiomatization of {\tt if-then-else} over halting programs and tests. They also modelled composition of functions and of functions with predicates and further showed that the more natural setting of only considering composition of functions would not admit a finite axiomatization.

The work listed above mainly focus on halting tests (by assuming them to be of Boolean type) and halting programs. A natural interest in this context is to study non-halting tests and programs. On these lines as well considerable work has been done, besides the work in \cite{mccarthy63} and \cite{igarashi71}.

In \cite{bloom83} Bloom and Tindell studied four versions of {\tt if-then-else} along with the equality test. In two cases they considered the halting scenario whilst in the other two they modeled possibly non-halting programs and tests. They provided an equationally complete proof system for each such framework while noting that none of the classes formed an equational class. In order to obtain similar results in the context of functional programming languages that have user-definable data types, in \cite{guessarian87}, Guessarian and Mesegeuer extended the proof system of \cite{bloom83} to heterogeneous algebras that have extra operations, predicates and equations. Another extension of \cite{bloom83} was by Mekler and Nelson in \cite{mekler87}. In this work the authors expanded the algebras in some equational class $K$ by adding the {\tt if-then-else} operation and found axioms for the equational class $K^*$ generated by these algebras. They also showed that the equational theory for $K^*$ is decidable if the word problem for $K$ is decidable. On a slightly different track, Manes in \cite{manes90} gave a transformational characterisation for {\tt if-then-else} where the tests are Boolean but the functions on which they act could be non-halting. Further, in \cite{manes93}, Manes considered {\tt if-then-else} algebras over Boolean algebras, $C$-algebras and adas. Here $C$-algebras and adas are algebras of non-halting tests, generalizing Boolean algebras to three-valued logics.

While there are several studies (cf. \cite{belnap70}, \cite{bergstra95}, \cite{bochvar38}, \cite{heyting34}, \cite{kleene38}, \cite{lukasiewicz20}) on extending two-valued Boolean logic to three-valued logic, McCarthy's logic (cf. \cite{mccarthy63}) models the lazy evaluation exhibited by programming languages that evaluate expressions in sequential order, from left to right. In \cite{guzman90} Guzm\'{a}n and Squier gave a complete axiomatization of McCarthy's three-valued logic and called the corresponding algebra a $C$-algebra, or the algebra of conditional logic. While studying {\tt if-then-else} algebras, in \cite{manes93}, Manes defined an {\em ada} (Algebra of Disjoint Alternatives) which is essentially a $C$-algebra equipped with an oracle for the halting problem.

Recently, in \cite{jackson15} Jackson and Stokes studied the algebraic theory of computable functions, which can be viewed as possibly non-halting programs, equipped with composition, {\tt if-then-else} and {\tt while-do}. In this work they assumed that the tests form a Boolean algebra. Further, they demonstrated how an algebra of non-halting tests could be constructed from Boolean tests in their setting. Jackson and Stokes proposed an alternative approach by considering an abstract collection of non-halting tests as in \cite{manes93} and posed the following problem:

{\em Characterize the algebras of computable functions associated with an abstract $C$-algebra of non-halting tests.}

In this paper, we attempt to address the problem by adopting the approach of Jackson and Stokes in \cite{jackson09}. To that aim, we define the notion of a {\em $C$-set} through which we provide a complete axiomatization for {\tt if-then-else} over a class of possibly non-halting programs and tests, where tests are drawn from an ada. The paper has been organized as follows. The necessary background material is provided in Section 1. In Section 2, we introduce the notion of $C$-sets and give a few properties of $C$-sets. Targeting our goal, Section 3 is dedicated to provide a subdirect representation of $C$-sets over adas. Further, in Section 4 we give a complete axiomatization for $C$-sets over adas equipped with the equality test, called {\em agreeable $C$-sets}.  A brief conclusion with possible extensions of this work are discussed in Section 5.

\section{Preliminaries}


In this section, we shall list definitions and results that will be useful to us. In \cite{jackson09} Jackson and Stokes considered the notion of a $B$-set in order to study the theory of halting programs equipped with the operation of {\tt if-then-else}.

\begin{definition}
 Let $\langle Q, \vee, \wedge, \neg, T, F \rangle$ be a Boolean algebra and $S$
be a set. A \emph{$B$-set} is a pair $(S, Q)$, equipped with a function, called \emph{$B$-action} $\eta : Q \times S \times S \rightarrow
S$, where $\eta(\alpha, a, b)$ is denoted by $\alpha[a, b]$, read ``{\tt if $\alpha$ then $a$ else $b$}", that satisfies the following axioms for all $\alpha, \beta \in Q$ and $a, b, c \in S$:

\begin{align}
  \alpha[a, a] & = a \label{B1} \\
   \alpha[\alpha[a, b], c] & = \alpha[a, c] \label{B2} \\
   \alpha[a, \alpha[b, c]] & = \alpha[a, c] \label{B3} \\
   F[a, b] & = b \label{B4} \\
   \neg \alpha[a, b] & = \alpha[b, a] \label{B5} \\
   (\alpha \wedge \beta)[a, b] & = \alpha[\beta[a, b], b] \label{B6}
\end{align}
\end{definition}

We recall the following examples from \cite{jackson09}.

\begin{example}
 For any Boolean algebra $Q$, the pair $(Q, Q)$ is a $B$-set with the following action for all $\alpha, \beta, \gamma \in Q$:
 $$\alpha [\beta, \gamma] = (\alpha \wedge \beta) \vee (\neg \alpha \wedge
\gamma).$$
\end{example}

\begin{example}
 Consider the two-element Boolean algebra $\mathbb{2}$ with the universe $\{T, F \}$.  For any set $S$, the pair $(S, \mathbb{2})$ is a $B$-set with the following action for all $a, b \in S$:
 \begin{align*}
  T[a, b] & = a, \\
  F[a, b] & = b. \newline
 \end{align*}
 These $B$-sets are called \emph{basic} $B$-sets.
\end{example}

\begin{notation}
Let $X$ and $Y$ be two sets. The set of all functions from $X$ to $Y$ will be denoted by $Y^{X}$.  The set of all total functions over $X$ will be denoted by $\mathcal{T}(X)$.
\end{notation}

\begin{example}
 For any set $X$, the pair $(\mathcal{T}(X), \mathbb{2}^{X})$ is a $B$-set with the following action for all $\alpha \in \mathbb{2}^{X}$ and $g, h \in \mathcal{T}(X)$:
  \begin{equation*}
  \alpha[g, h](x) =
  \begin{cases}
   g(x), & \text{ if } \alpha(x) = T; \\
   h(x), & \text{ if } \alpha(x) = F.
  \end{cases}
 \end{equation*}
\end{example}

In \cite{jackson09} Jackson and Stokes showed that every $B$-set can be represented in terms of basic $B$-sets.

\begin{theorem}[\cite{jackson09}] \label{BsetSubdir}
 Every $B$-set is a subdirect product of basic $B$-sets.
\end{theorem}

This tells us that studying the identities satisfied by the subclass of basic $B$-sets suffices to understand those satisfied by the entire class of $B$-sets. Checking the validity of any identity in a basic $B$-set involves merely checking the respective values for {\tt true} and {\tt false} and is thus far simpler than checking the same in an arbitrary $B$-set. Further, in \cite{jackson09}, they model the equality test based on the assumption that the tests arise from a Boolean algebra and that the functions are halting.

\begin{definition}
 A $B$-set $(S, B)$ is said to be \emph{agreeable} if it is equipped with an operation $* : S \times S \rightarrow B$ satisfying the following axioms for all $s, t, u, v \in S$ and $\alpha \in B$:

 \begin{align}
  s * s & = T \label{AB1} \\
  (s * t)[s, t] & = t \label{AB2} \\
  \alpha[s, t] * \alpha[u, v] & = \alpha[s * u, t * v] \label{AB3}
 \end{align}
\end{definition}

The following are examples of agreeable $B$-sets.

\begin{example}
 The pair $(\mathcal{T}(X), \mathbb{2}^{X})$ is an agreeable $B$-set with the operation $*$ defined as follows for all $f, g \in \mathcal{T}(X)$:
 \begin{equation*}
  (f * g)(x) =
  \begin{cases}
   T, & \text{ if } f(x) = g(x); \\
   F, & \text{ otherwise.}
  \end{cases}
 \end{equation*}
\end{example}

\begin{example}
 Let $S$ be any set. The pair $(S, \mathbb{2})$ is an agreeable $B$-set under the operation $*$ defined in the following manner for all $s, t \in S$:
 \begin{equation*}
  s * t =
  \begin{cases}
   T, & \text{ if } s = t; \\
   F, & \text{ otherwise.}
  \end{cases}
 \end{equation*}
 These $B$-sets are called \emph{basic agreeable $B$-sets}.
\end{example}

Jackson and Stokes proved the following result.

\begin{theorem}[\cite{jackson09}] \label{AgreeableBsetSubdir}
 Every agreeable $B$-set is a subdirect product of basic agreeable $B$-sets.
\end{theorem}

In \cite{kleene52} Kleene discussed various three-valued logics that are extensions of Boolean logic. McCarthy first studied the three-valued non-commutative logic in the context of programming languages in \cite{mccarthy63}. This is the non-commutative regular extension of Boolean logic to three truth values. Here the third truth value $U$ denotes the {\tt undefined} state which is attained when a test diverges. In this new context, the evaluation of expressions is carried out sequentially from left to right, mimicking that of a majority of programming languages. The complete axiomatization for the class of algebras associated with this logic was given by Guzm\'{a}n and Squier in \cite{guzman90}. They called the algebra associated with this logic as a \emph{$C$-algebra}.  We shall denote an arbitrary $C$-algebra by $M$.

\begin{definition}
 A \emph{$C$-algebra} is an algebra $\langle M, \vee, \wedge, \neg \rangle$ of
type $(2, 2, 1)$, which satisfies the following axioms for all $\alpha, \beta, \gamma \in M$:

\begin{align}
  \neg \neg \alpha & = \alpha \label{C1} \\
   \neg (\alpha \wedge \beta) & = \neg \alpha \vee \neg \beta \label{C2} \\
   (\alpha \wedge \beta) \wedge \gamma & = \alpha \wedge (\beta \wedge \gamma) \label{C3} \\
   \alpha \wedge (\beta \vee \gamma) & = (\alpha \wedge \beta) \vee (\alpha \wedge \gamma) \label{C4} \\
   (\alpha \vee \beta) \wedge \gamma & = (\alpha \wedge \gamma) \vee (\neg \alpha \wedge \beta \wedge \gamma) \label{C5} \\
   \alpha \vee (\alpha \wedge \beta) & = \alpha \label{C6} \\
   (\alpha \wedge \beta) \vee (\beta \wedge \alpha) & = (\beta \wedge \alpha) \vee (\alpha \wedge \beta) \label{C7}
\end{align}

\end{definition}

\begin{example}
Every Boolean algebra is a $C$-algebra. In particular, $\mathbb{2}$ is a $C$-algebra.
\end{example}

\begin{example}
 Let $\mathbb{3}$ denote the $C$-algebra with the universe $\{ T, F, U \}$ and the following operations. This is, in fact, McCarthy's three-valued logic.
 \begin{center}
  \begin{tabular}{c|c}
  $\neg$ & \\
  \hline
  $T$ & $F$ \\
  $F$ & $T$ \\
  $U$ & $U$
 \end{tabular}
 \quad
 \begin{tabular}{c|ccc}
  $\wedge$ & $T$ & $F$ & $U$ \\
  \hline
  $T$ & $T$ & $F$ & $U$ \\
  $F$ & $F$ & $F$ & $F$ \\
  $U$ & $U$ & $U$ & $U$
 \end{tabular}
 \quad
 \begin{tabular}{c|ccc}
  $\vee$ & $T$ & $F$ & $U$ \\
  \hline
  $T$ & $T$ & $T$ & $T$ \\
  $F$ & $T$ & $F$ & $U$ \\
  $U$ & $U$ & $U$ & $U$
 \end{tabular}
 \end{center}
 \quad \newline
\end{example}

Guzm\'{a}n and Squier showed the following result.

\begin{theorem}[\cite{guzman90}] \label{SubdirIrredCAlg}
$\mathbb{3}$ and $\mathbb{2}$ are the only subdirectly irreducible $C$-algebras. Hence, every $C$-algebra is a subalgebra of a product of copies of $\mathbb{3}$.
\end{theorem}

This tells us that for any set $X$, $\mathbb{3}^{X}$ is a $C$-algebra with the operations defined point-wise. In fact, Guzm\'{a}n and Squier showed that elements of $\mathbb{3}^{X}$ along with the $C$-algebra operations may be viewed in terms of \emph{pairs of sets}. This is a pair $(A, B)$ where $A, B \subseteq X$ and $A \cap B = \emptyset$. Akin to the well-known correlation between $\mathbb{2}^{X}$ and the power set of $X$, $\powerset(X)$, for any element $\alpha \in \mathbb{3}^{X}$, associate the pair of sets $(A, B)$ where $A = \{ x \in X : \alpha(x) = T \}$ and $B = \{ x \in X : \alpha(x) = F \}$. Conversely, for any pair of sets $(A, B)$ where $A, B \subseteq X$ and $A \cap B = \emptyset$ associate the function $\alpha$ where $\alpha(x) = T$ if $x \in A$, $\alpha(x) = F$ if $x \in B$ and $\alpha(x) = U$ otherwise. With this correlation, the operations can be expressed as follows:

\begin{align*}
 \neg (A_{1}, A_{2}) & = (A_{2}, A_{1}) \\
 (A_{1}, A_{2}) \wedge (B_{1}, B_{2}) & = (A_{1} \cap B_{1}, A_{2} \cup (A_{1} \cap B_{2})) \\
 (A_{1}, A_{2}) \vee (B_{1}, B_{2}) & = ((A_{1} \cup (A_{2} \cap B_{1}), A_{2} \cap B_{2}) \\
\end{align*}

\begin{remark}
 Considering a $C$-algebra $M$ as a subalgebra of $\mathbb{3}^{X}$, one may observe that $M_{\#} = \{ \alpha \in M : \alpha \vee \neg \alpha = T \}$ forms a Boolean algebra under the induced operations.
\end{remark}

\begin{notation}
 A $C$-algebra with $T, F, U$ is meant a $C$-algebra with nullary operations $T, F, U$, where $T$ is the (unique) left-identity (and right-identity) for $\wedge$, $F$ is the (unique) left-identity (and right-identity) for $\vee$ and $U$ is the (unique) fixed point for $\neg$. Note that $U$ is also a left-zero for both $\wedge$ and $\vee$ while $F$ is a left-zero for $\wedge$.
\end{notation}

There is an important subclass of the variety of $C$-algebras. In \cite{manes93} Manes introduced the notion of algebra of disjoint alternatives, in short \emph{ada}, which is a $C$-algebra equipped with an oracle for the halting problem. He showed that the category of adas are equivalent to that of Boolean algebras. The $C$-algebra $\mathbb{3}$ is not functionally-complete. However, $\mathbb{3}$ is functionally-complete when treated as an ada. In fact, the variety of adas is generated by the ada $\mathbb{3}$.

\begin{definition}
 An \emph{ada} is a $C$-algebra $M$ with $T, F, U$ equipped with an additional unary operation $(\text{ })^{\downarrow}$ subject to the following equations for all $\alpha, \beta \in M$:
 \begin{align}
  F^{\downarrow} & = F \label{A1} \\
  U^{\downarrow} & = F \label{A2} \\
  T^{\downarrow} & = T \label{A3} \\
  \alpha \wedge \beta^{\downarrow} & = \alpha \wedge (\alpha \wedge \beta)^{\downarrow} \label{A4} \\
  \alpha^{\downarrow} \vee \neg (\alpha^{\downarrow}) & = T \label{A5} \\
  \alpha & = \alpha^{\downarrow} \vee \alpha \label{A6}
 \end{align}
\end{definition}

\begin{example}
 The three-element $C$-algebra $\mathbb{3}$ with the unary operation $(\text{ })^{\downarrow}$ defined as follows forms an ada.
 \begin{align*}
  T^{\downarrow} & = T \\
  U^{\downarrow} & = F = F^{\downarrow}
 \end{align*}
We shall also use $\mathbb{3}$ to denote this ada. One may easily resolve the notation overloading -- whether $\mathbb{3}$ is a $C$-algebra or an ada --  depending on the context.

\end{example}

In \cite{manes93} Manes showed that the three-element ada $\mathbb{3}$ is the only subdirectly irreducible ada. Thus for any set $X$, $\mathbb{3}^{X}$ is an ada with operations defined pointwise. Note that the three element ada $\mathbb{3}$ is also simple. Manes also showed the following result.

\begin{proposition}[\cite{manes93}]
Let $A$ be an ada. Then $A^{\downarrow} = \{ \alpha^{\downarrow} : \alpha \in A \}$ forms a Boolean algebra under the induced operations.
\end{proposition}

\begin{remark}
 In fact, $A^{\downarrow} = A_{\#}$. Also, $A^{\downarrow} = \{ \alpha \in A : \alpha^{\downarrow} = \alpha \}.$
\end{remark}

Further, as outlined in the following remark, Manes established that the category of adas and the category of Boolean algebras are equivalent.

\begin{remark}[\cite{manes93}] \label{RemarkStone}
Let $B$ be a Boolean algebra. In view of Stone's representation of Boolean algebras, suppose $B$ is a subalgebra of $\mathbb{2}^{X}$ for some set $X$. Consider the subalgebra $B^{\star}$ of the ada $\mathbb{3}^{X}$ with the universe $B^{\star} = \{(P, Q) : P \cap Q = \emptyset\}$ given in terms of pairs of subsets of $X$. Note that the map $B \mapsto (B^{\star})_{\#}$ is a Boolean isomorphism. Similarly, for an ada $A$, the map $A \mapsto (A_{\#})^{\star}$ is an ada isomorphism. Hence, the functor based on the aforesaid assignment establishes that the category of adas and the category of Boolean algebras are equivalent.
\end{remark}

\begin{notation}
Let $X$ be a set and $\bot \notin X$. The pointed set $X \cup \{ \bot \}$ with base point $\bot$ is denoted by $X_{\bot}$.
The set of all total functions on $X_{\bot}$ which fix $\bot$ is denoted by $\mathcal{T}_{o}(X_{\bot})$, i.e. $\mathcal{T}_{o}(X_{\bot}) = \{f \in \mathcal{T}(X_{\bot}) \; : \; f(\bot) = \bot\}$.
\end{notation}

\section{$C$-sets} \label{Csets}

In this section we introduce the notion of a $C$-set to study an axiomatization of {\tt if-then-else} that includes the models of possibly non-halting programs and tests. The concept of $C$-sets is an extension of the notion of $B$-sets, wherein the tests are drawn from a $C$-algebra instead of a Boolean algebra, and includes a non-halting or {\tt error} state.

\begin{definition}
 Let $S_{\bot}$ be a pointed set with base point $\bot$ and $M$ be a $C$-algebra with $T, F, U$. The pair $(S_{\bot}, M)$ is called a \emph{$C$-set} if it is equipped with an action \[\_\; [\_\; , \_] : M \times S_{\bot} \times S_{\bot} \rightarrow S_{\bot}\] that satisfies the following axioms for all $\alpha, \beta \in M$ and $s, t, u, v \in S_{\bot}$:
 \begin{align}
  U[s, t] & = \bot \label{EC1} & \text{($U$-axiom)} \\
  F[s, t] & = t \label{EC6} & \text{($F$-axiom)} \\
  (\neg \alpha)[s, t] & = \alpha[t, s] \label{EC5} & \text{($\neg$-axiom)} \\
  \alpha[\alpha[s, t], u] & = \alpha[s, u]  \label{EC3} & \text{(positive redundancy)} \\
  \alpha[s, \alpha[t, u]] & = \alpha[s, u] \label{EC4} & \text{(negative redundancy)} \\
  (\alpha \wedge \beta)[s, t] & = \alpha[\beta[s, t], t] \label{EC7} & \text{($\wedge$-axiom)} \\
  \alpha[\beta[s, t], \beta[u, v]] & = \beta[\alpha[s, u], \alpha[t, v]] \label{EC2} & \text{(premise interchange)} \\
  \alpha[s, t] = \alpha[t, t] & \Rightarrow (\alpha \wedge \beta)[s, t] = (\alpha \wedge \beta)[t, t] \label{EC8} & \text{($\wedge$-compatibility)}
 \end{align}
\end{definition}

\begin{remark}\label{v-axiom}
In view of \eqref{C1} and \eqref{C2} of $C$-algebras and \eqref{EC5} and \eqref{EC7} of $C$-sets, we have the following property in $C$-sets.
\begin{align}
  (\alpha \vee \beta)[s, t] & = \alpha[s, \beta[s, t]] \label{VA} & \text{($\vee$-axiom)}
  \end{align}
\end{remark}

We shall now present the intuition behind the notion of $C$-set and its axioms with respect to program constructs. In order to include the possibility of non-halting tests, we shall assume that the tests form a $C$-algebra. A test diverges at a given input if the output evaluates to $U$, {\tt undefined}. When a test diverges or if the program throws up an {\tt error} or does not halt, we shall say that the program evaluates to $\bot$. Thus a pointed set $S_{\bot}$ models the set of states and base point $\bot$ serves to denote the {\tt error} state.

The $U$-axiom \eqref{EC1} essentially encapsulates the real-world requirement that if a test diverges, the output should fall in the {\tt error} state. The $F$-axiom \eqref{EC6} is natural as when the test is {\tt false}, the {\tt then} part of the {\tt if-then-else} construct is executed. The $\neg$-axiom \eqref{EC5} simply states that executing {\tt if} \emph{not $P$} {\tt then} $f$ {\tt else} $g$ is the same as executing {\tt if} \emph{$P$} {\tt then} $g$ {\tt else} $f$. The axioms of positive redundancy and negative redundancy \eqref{EC3} and \eqref{EC4} encapsulate the \emph{cascading} nature of {\tt if-then-else}. The $\wedge$-axiom \eqref{EC7} states that evaluating test $P$ \emph{AND} $Q$ and then executing $f$ else $g$ works in exactly the same way as evaluating $P$ first, which if {\tt true}, executing {\tt if} $Q$ {\tt then} $f$ {\tt else} $g$, and if {\tt false} then simply executing $g$. The axiom of premise interchange \eqref{EC2} serves as a \emph{switching} law. This states that the behaviour of the program where $P$ is evaluated first and $Q$ is a test situated within both the branches of the main {\tt if-then-else}, is exactly the same as evaluating $Q$ first with $P$ in each branch, on suitably interchanging the programs situated at the leaves. The last axiom of $\wedge$-compatibility \eqref{EC8} loosely means that if $f$ and $g$ agree with regards to some domain, then they will agree on any subdomain.

\begin{example} \label{e-m-m}
Let $M$ be a $C$-algebra with $T, F, U$. By treating $M$ as a pointed set with base point $U$, the pair $(M, M)$ is a $C$-set under the following action for all $\alpha, \beta, \gamma \in M$:
 $$\alpha[\beta, \gamma] = (\alpha \wedge \beta) \vee (\neg \alpha \wedge \gamma).$$
Hereafter, the action of the $C$-set $(M, M)$ will be denoted by double brackets $\_\; \llbracket \_ \;, \_ \rrbracket$. For verification of the axioms \eqref{EC1} -- \eqref{EC8} refer to Appendix \ref{VerCAlgCset}.
\end{example}

We now present the motivating example of $C$-sets. Since the natural models of possibly non-halting programs are partial functions, we consider the model $\mathcal{T}_{o}(X_{\bot})$ in view of the following one-to-one correspondence between  $\mathcal{T}_{o}(X_{\bot})$ and the set of partial functions on a set $X$. Each partial function $f$ on $X$ is represented by the total function $f' \in \mathcal{T}_{o}(X_{\bot})$ where $f'(x) = f(x)$ when $x$ is in the domain of $f$, and maps to $\bot$ otherwise. Conversely, each $g \in \mathcal{T}_{o}(X_{\bot})$ is represented by the partial function $g''$ over $X$ where $g''(x) = g(x)$ when $x \in X$ and $g(x) \neq \bot$, and is not defined elsewhere. The model $\mathcal{T}_{o}(X_{\bot})$ can be seen as a $C$-set under the action of the $C$-algebra $\mathbb{3}^{X}$ as per the following example.

\begin{example} \label{ExampleFunctionalCset}
Consider $\mathcal{T}_{o}(X_{\bot})$ as a pointed set with base point $\zeta_{\bot}$, the constant function taking the value $\bot$. The pair $\big( \mathcal{T}_{o}(X_{\bot}), \mathbb{3}^{X}\big)$ is a $C$-set with the following action for all $f, g \in \mathcal{T}_{o}(X_{\bot})$ and $\alpha \in \mathbb{3}^{X}$:

\begin{equation} \label{FunctionalAction}
 \alpha[f, g](x) =
 \begin{cases}
  f(x), & \text{ if } \alpha(x) = T; \\
  g(x), & \text{ if } \alpha(x) = F; \\
  \bot, & \text{ otherwise. }
 \end{cases}
\end{equation}
Note that the execution of the first two cases, $\alpha(x) \in \{ T, F \}$ demands that $x \in X$ as $\alpha \in \mathbb{3}^{X}$. Using the pairs of sets representation of $\mathbb{3}^{X}$, one may routinely verify that the pair $\big( \mathcal{T}_{o}(X_{\bot}), \mathbb{3}^{X}\big)$ satisfies the axioms \eqref{EC1} -- \eqref{EC8}. These $C$-sets will be called \emph{functional $C$-sets}. For verification of the axioms \eqref{EC1} -- \eqref{EC8} refer to Appendix \ref{VerFunctionalCset}.
\end{example}

\begin{example}
Consider $S_{\bot}^{X}$, the set of all functions from $X$ to $S_\bot$, as a pointed set with base point $\zeta_{\bot}$. The pair $\big(S_{\bot}^{X}, \mathbb{3}^{X}\big)$ is a $C$-set under the action given in \eqref{FunctionalAction}, where $f, g \in S_{\bot}^{X}$ and $\alpha \in \mathbb{3}^{X}$. The axioms \eqref{EC1} -- \eqref{EC8} can be verified on the same lines as in \eref{ExampleFunctionalCset}.
\end{example}

\begin{example}\label{tot-c-set}
Consider $\mathcal{T}(X_{\bot})$, the set of all total functions on $X_\bot$, as a pointed set with base point $\zeta_{\bot}$. The pair $\big(\mathcal{T}(X_{\bot}), \mathbb{3}^{X}\big)$ is a $C$-set under the action given in \eqref{FunctionalAction}, where $f, g \in \mathcal{T}(X_{\bot})$ and $\alpha \in \mathbb{3}^{X}$. The axioms \eqref{EC1} -- \eqref{EC8} can be verified on the same lines as in \eref{ExampleFunctionalCset}.
\end{example}

We believe that the $C$-set given in \eref{tot-c-set} does not occur naturally in the context of programs as this would include elements that terminate even when the input diverges, i.e. the input is $\bot$.

We now present a fundamental example of a $C$-set, where we only consider the basic tests, {\tt true, false, undefined}.

\begin{example} \label{ExampleBasicCset}
Let $S_{\bot}$ be a pointed set with base point $\bot$. The pair $(S_{\bot}, \mathbb{3})$ is a $C$-set with respect to the following action for all $a, b \in S_\bot$ and $\alpha \in \mathbb{3}$:
 \begin{equation*}
  \alpha[a, b] =
  \begin{cases}
   a, & \text{ if } \alpha = T; \\
   b, & \text{ if } \alpha = F; \\
   \bot, & \text{ if } \alpha = U.
  \end{cases}
 \end{equation*}
 These $C$-sets are called \emph{basic $C$-sets}. For verification of the axioms \eqref{EC1} -- \eqref{EC8} refer to Appendix \ref{VerBasicCset}.
\end{example}

Henceforth, unless explicitly mentioned otherwise, an arbitrary $C$-set is always denoted by $(S_{\bot}, M)$. In the remainder of this section, we shall prove certain properties of $C$-sets.

\begin{proposition} \label{PropCset}
The following statements hold for all $\alpha, \beta \in M$ and $s, t, r \in S_{\bot}$:
 \begin{enumerate}[\rm(i)]
  \item $\alpha[\bot, \bot] = \bot$.
  \item If $\alpha[s, u] = \alpha[t, q]$ for some $u, q \in S_{\bot}$ then $\alpha[s, v] = \alpha[t, v]$ for all $v \in S_{\bot}$.
  \item If $\alpha[s, u] = \alpha[r, r]$ for some $u \in S_{\bot}$ then $\alpha[s, r] = \alpha[r, r]$.
  \item If $\alpha[s, u] = \alpha[t, u]$ for some $u \in S_{\bot}$ then $\alpha[s, v] = \alpha[t, v]$ for all $v \in S_{\bot}$.
  \item If $\alpha[s, t] = \alpha[t, t]$ then $(\beta \wedge \alpha)[s, t] = (\beta \wedge \alpha)[t, t]$.
 \end{enumerate}
\end{proposition}

\begin{proof}$\;$
 \begin{enumerate}[(i)]
  \item Using \eqref{EC1} and \eqref{EC2}, $\alpha[\bot, \bot] = \alpha[U[\bot, \bot], U[\bot, \bot]] = U[\alpha[\bot, \bot], \alpha[\bot, \bot]] = \bot$.
  \item Using \eqref{EC3}, $\alpha[s, v] = \alpha[ \alpha[s, u], v] = \alpha[ \alpha[t, q], v] = \alpha[t, v]$.
  \item Using \pref{PropCset}(ii), putting $t = q = v = r$, $\alpha[s, r] = \alpha[r, r]$.
  \item Using \pref{PropCset}(ii), putting $q = u$, $\alpha[s, v] = \alpha[t, v]$.
  \item Using \eqref{EC7}, $(\beta \wedge \alpha)[s, t] = \beta[\alpha[s, t], t] = \beta[\alpha[t, t], t] = (\beta \wedge \alpha)[t, t]$.
 \end{enumerate}
\end{proof}

\begin{remark}$\;$
 \begin{enumerate}[(i)]
 \item The $C$-set axioms from \eqref{EC6} to \eqref{EC7} are the same as the ones in the definition of $B$-set. In view of \eqref{EC1}, the only $B$-set axiom that does not carry over in the context of $C$-sets is \eqref{B1}.
  \item It is routine to verify that the axiom of premise interchange \eqref{EC2} holds in a basic $B$-set. Hence, in view of \tref{BsetSubdir}, axiom \eqref{EC2} holds in all $B$-sets.
  \item Following the proof given in \pref{PropCset}(v) and using the commutativity of $\wedge$ in the context of $B$-sets, it can be observed that the axiom of $\wedge$-compatibility \eqref{EC8} holds in $B$-sets.
 \end{enumerate}
\end{remark}

\begin{proposition} \label{PropBset}
 For each $\alpha \in M_{\#}$ and $s \in S_{\bot}$, we have $\displaystyle \alpha[s, s] = s$.
\end{proposition}

\begin{proof}
 Let $\alpha \in M_{\#}$ and $s \in S_{\bot}$.
 \begin{align*}
  s & = T[s, s] & \text{ from } \eqref{EC5}, \eqref{EC6} \\
    & = (\alpha \vee (\neg \alpha))[s, s] & \text{ since } \alpha \in M_{\#} \\
    & = \alpha[s, (\neg \alpha)[s, s]] & \text{ from } \eqref{VA}\\
    & = \alpha[s, \alpha[s, s]] & \text{ from } \eqref{EC5} \\
    & = \alpha[s, s] & \text{ from } \eqref{EC4}
 \end{align*}
\end{proof}

In view of \pref{PropBset}, the axiom \eqref{B1} of $B$-sets holds for the elements of Boolean algebra $M_{\#}$. Hence, we have the following corollary.

\begin{corollary} \label{CorBset}
The pair $(S_{\bot}, M_{\#})$ is a $B$-set.
\end{corollary}

\begin{remark}
The proof of \pref{PropBset} also shows us that axiom \eqref{B1} is redundant in the definition of a $B$-set.
\end{remark}

\section{Representation of $C$-sets} \label{RepresentationCset}

With the aim of studying structural properties of $C$-sets, in this section we obtain a subdirect representation of $C$-sets in which the $C$-algebras are adas. Note that, except in \eref{e-m-m}, the $C$-algebras in all other examples of $C$-sets given in \sref{Csets} are adas.

Let $(S_\bot, M)$ be a $C$-set, where $M$ is an ada. In the main theorem (\tref{CsetSubdir}) of this section, we obtain a subdirect representation of $(S_\bot, M)$ through various results presented hereafter. In those results, we consistently use $\alpha, \beta, \gamma$ for the elements of $M$, and $r, s, t, u, v$ for the elements of $S_\bot$.

\begin{proposition} \label{PropCsetarrow}
 If $\alpha[s, t] = \alpha[t, t]$, then $\alpha^{\downarrow}[s, t] = \alpha^{\downarrow}[t, t]$.
\end{proposition}

\begin{proof}
Using \eqref{A6} and \eqref{VA}, $\alpha[s, t] = (\alpha^{\downarrow} \vee \alpha)[s, t] = \alpha^{\downarrow}[s, \alpha[s, t]] = \alpha^{\downarrow}[s, \alpha[t, t]]$. On the other hand, observe that $\alpha[s, t] = \alpha[t, t] = (\alpha^{\downarrow} \vee \alpha)[t, t] = \alpha^{\downarrow}[t, \alpha[t, t]]$ so that $\alpha^{\downarrow}[s, \alpha[t, t]] = \alpha^{\downarrow}[t, \alpha[t, t]]$.  Consequently, we have  $\alpha^{\downarrow}[s, t] = \alpha^{\downarrow}[t, t]$ by \pref{PropCset}(iv).
\end{proof}

Considering the $C$-set as a two-sorted algebra, in the following, we define its congruence.

\begin{definition}
A \emph{congruence} of a $C$-set is a pair $(\sigma, \tau)$, where $\sigma$ is an equivalence relation on $S_{\bot}$ and $\tau$ is a congruence on the ada $M$ such that
$$(s, t), (u, v) \in \sigma \; \text{ and }\; (\alpha, \beta) \in \tau \Rightarrow (\alpha[s, u], \beta[t, v]) \in \sigma.$$
\end{definition}

\begin{notation}
Under an equivalence relation $\sigma$ on a set $A$, the equivalence class of an element $p \in A$ will be denoted by $\overline{p}^{{\sigma}}$. Within a given context, if there is no ambiguity, we may simply denote the equivalence class by $\overline{p}$.
\end{notation}

In order to give a subdirect representation of the $C$-set $(S_\bot, M)$, we shall consider the collection of all maximal congruences on the ada $M$ so that for each such congruence $\theta$, we have $M/\theta \cong \mathbb{3}$. We shall produce an equivalence relation $E_\theta$ on $S_\bot$ such that $(E_\theta, \theta)$ is a congruence on $(S_{\bot}, M)$ for each $\theta$, and the intersection of the collection of congruences $(E_\theta, \theta)$ is trivial. Thus, $(S_\bot, M)$ is a subdirect product of basic $C$-sets $(S_\bot/E_\theta, M/\theta)$.

\begin{definition} \label{DefinitionETheta}
 For each maximal congruence $\theta$ on $M$, we define a relation on $S_{\bot}$ by
 $$E_{\theta} = \{ (s, t) \in S_{\bot} \times S_{\bot} : \;  \beta[s, t] = \beta[t, t]\; \text{ for some } \beta \in \overline{T}^{\theta}\}.$$
\end{definition}

\begin{lemma} \label{LemmaEquiv}
 The relation $E_{\theta}$ is an equivalence on $S_{\bot}$.
\end{lemma}

\begin{proof}
Since $T[s, s] = T[s, s]$ and $T \in \overline{T}^{\theta}$, we have $(s, s) \in E_{\theta}$ so that the binary relation $E_{\theta}$ on $S_{\bot}$ is reflexive.

For symmetry, let $(s, t) \in E_{\theta}$. Then there exists $\beta \in \overline{T}^{\theta}$ such that $\beta[s, t] = \beta[t, t]$. Using \eqref{EC3}, we have $\beta[t, s] = \beta[\beta[t, t], s] = \beta[\beta[s, t], s] = \beta[s, s]$ so that $(t, s) \in E_{\theta}$.

Let $(s, t), (t, r) \in E_{\theta}$. Then there exist $\alpha, \beta \in \overline{T}^{\theta}$, such that $\alpha[s, t] = \alpha[t, t]$ and $\beta[t, r] = \beta[r, r]$. As $\theta$ is a congruence on $M$,  $(\alpha, T), (\beta, T) \in \theta$  implies $(\alpha \wedge \beta, T) \in \theta$ so that $(\alpha \wedge \beta) \in \overline{T}^{\theta}$.  Note that
\[\begin{array}{rcll}
  (\alpha \wedge \beta)[s, r] &=& (\alpha \wedge \beta)[(\alpha \wedge \beta)[s, t], r] & \text{from}\; \eqref{EC3} \\
  &=& (\alpha \wedge \beta)[(\alpha \wedge \beta)[t, t], r] & \text{from}\; \alpha[s, t] = \alpha[t, t] \; \text{and}\; \eqref{EC8} \\
  &=& (\alpha \wedge \beta)[t, r] & \text{from}\; \eqref{EC3} \\
  &=&  (\alpha \wedge \beta)[r, r] & \text{from}\; \beta[t, r] = \beta[r, r] \text{ and \pref{PropCset}(v)}\; .
\end{array}\]
Hence $(s, r) \in E_{\theta}$ so that $E_{\theta}$ is transitive.
\end{proof}

\begin{remark} \label{RemarkQuotient}
Note that, as $\theta$ is a maximal congruence on ada $M$, $M / \theta$ must be simple, i.e., $M / \theta \cong \mathbb{3}$. Further, the quotient set  $S_{\bot} / E_{\theta}$ can be treated as a pointed set with base point $\overline{\bot}$. Thus $(S_{\bot} / E_{\theta}, M / \theta)$ is a basic $C$-set under the action \\
 \begin{equation*}
  \overline{\alpha}^{_{\theta}} [ \overline{s}^{_{E_{\theta}}}, \overline{t}^{_{E_{\theta}}} ] =
  \begin{cases}
   \overline{s}^{_{E_{\theta}}}, & \text{ if } \alpha \in \overline{T}^{\theta}; \\
   \overline{t}^{_{E_{\theta}}}, & \text{ if } \alpha \in \overline{F}^{\theta}; \\
   \overline{\bot}^{_{E_{\theta}}}, & \text{ if } \alpha \in \overline{U}^{\theta}. \\
  \end{cases}
 \end{equation*}
\end{remark}

\begin{proposition} \label{AlphaBeta}
 For any $\alpha \in M$, $\beta = \neg (\alpha^{\downarrow} \vee (\neg \alpha)^{\downarrow}) \vee U$ satisfies $\beta \wedge \alpha = U$. Moreover, if $(\alpha, U) \in \theta$ then $(\beta, T) \in \theta$.
\end{proposition}

\begin{proof}
Since $\mathbb{3}$ is the only subdirectly irreducible ada, it is sufficient to check the validity of the identity $\beta \wedge \alpha = U$ in $\mathbb{3}$.

 \begin{enumerate}[]
  \item If $\alpha = T$, then $\beta = \neg ( T^{\downarrow} \vee F^{\downarrow}) \vee U = \neg ( T \vee F) \vee U = F \vee U = U$.
  \item If $\alpha = F$, then $\beta = \neg ( F^{\downarrow} \vee T^{\downarrow}) \vee U = \neg ( F \vee T) \vee U = F \vee U = U$.
  \item If $\alpha = U$, then $\beta = \neg ( U^{\downarrow} \vee U^{\downarrow}) \vee U = \neg ( F \vee F) \vee U = T \vee U = T$.
 \end{enumerate}
In all these three cases, it is straightforward to see that $\beta \wedge \alpha = U$.

Suppose $(\alpha, U) \in \theta$. Since $\theta$ is a congruence, we have $(\alpha^{\downarrow}, U^{\downarrow}) = (\alpha^{\downarrow}, F) \in \theta$. Also, we have $(\neg \alpha, \neg U) = (\neg \alpha, U) \in \theta$ so that $((\neg \alpha)^{\downarrow}, F) \in \theta$. Now, by substitution with respect to $\vee$, we have $(\alpha^{\downarrow} \vee (\neg \alpha)^{\downarrow}, F) \in \theta$.

This further implies
 $(\neg(\alpha^{\downarrow} \vee (\neg \alpha)^{\downarrow}), \neg F) = (\neg(\alpha^{\downarrow} \vee (\neg \alpha)^{\downarrow}), T) \in \theta$. However, since $(U, U) \in \theta$, we have $(\neg(\alpha^{\downarrow} \vee (\neg \alpha)^{\downarrow}) \vee U, T \vee U) = (\neg(\alpha^{\downarrow} \vee (\neg \alpha)^{\downarrow}) \vee U, T) \in \theta$. Hence, $(\beta, T) \in \theta$.
\end{proof}

\begin{proposition} \label{PropAlphaST}
For each $\alpha \in M$ and each $s, t \in S_{\bot}$, we have the following:
 \begin{enumerate}[\rm(i)]
  \item $(\alpha, T) \in \theta \Rightarrow (\alpha[s, t], s) \in E_{\theta}$.
  \item $(\alpha, F) \in \theta \Rightarrow (\alpha[s, t], t) \in E_{\theta}$.
  \item $(\alpha, U) \in \theta \Rightarrow (\alpha[s, t], \bot) \in E_{\theta}$.
 \end{enumerate}
\end{proposition}

\begin{proof}$\;$
 \begin{enumerate}[(i)]
  \item From \eqref{EC3}, we have $\alpha[\alpha[s, t], s] = \alpha[s, s]$. Hence $(\alpha[s, t], s) \in E_{\theta}$ as $\alpha \in \overline{T}^{\theta}$.

  \item Note that $(\alpha, F) \in \theta$ implies $(\neg \alpha, T) \in \theta$. Using \eqref{EC4} and \eqref{EC5}, $(\neg \alpha)[\alpha[s, t], t] = \alpha[t, \alpha[s, t]] = \alpha[t, t] = (\neg \alpha)[t, t]$. Thus $(\alpha[s, t], t) \in E_{\theta}$.

  \item If $(\alpha, U) \in \theta$, then by \pref{AlphaBeta}, $\beta = \neg (\alpha^{\downarrow} \vee (\neg \alpha)^{\downarrow}) \vee U \in \overline{T}^{\theta}$, and $\beta \wedge \alpha = U$. Note that

  \begin{align*}
   \beta[\alpha[s, t], t] & = (\beta \wedge \alpha)[s, t] & \text{ from } \eqref{EC7} \\
                          & = U[s, t] & \text{ from \pref{AlphaBeta}} \\
                          & = \bot & \text{ from } \eqref{EC1} \\
                          & = \beta[\bot, \bot] & \text{ from \pref{PropCset}(i)}.
  \end{align*}
Consequently, by \pref{PropCset}(iii), we have $\beta[\alpha[s, t], \bot] = \beta[\bot, \bot]$. Hence, $(\alpha[s, t], \bot) \in E_{\theta}$.
 \end{enumerate}

\end{proof}

\begin{lemma} \label{LemmaCong}
 The pair $(E_{\theta}, \theta)$ is a $C$-set congruence.
\end{lemma}

\begin{proof}
In view of \rref{RemarkQuotient}, $(S_{\bot} / E_{\theta}, M / \theta)$ is a basic $C$-set.
Consider the canonical maps $\nu_{1} : S_{\bot} \rightarrow S_{\bot} / E_{\theta}$, given by $\nu_{1}(s) = \overline{s}^{_{E_{\theta}}}$, and $\nu_{2} : M \rightarrow M / \theta \cong \mathbb{3}$, given by $\nu_{2}(\alpha) = \overline{\alpha}^{_{\theta}}$. We show that the pair $(\nu_{1}, \nu_{2})$ is a $C$-set homomorphism so that $\ker(\nu_{1}, \nu_{2}) = (E_{\theta}, \theta)$ is a $C$-set congruence.

It is straightforward to see that $\nu_{1}(\bot) = \overline{\bot}^{_{E_{\theta}}}$ and thus $\nu_{1}$ is a homomorphism of pointed sets. It is also clear that $\nu_{2}$ is a homomorphism of adas. Additionally, we require that $\nu_{1}(\alpha[s, t]) = (\nu_{2}(\alpha))[\nu_{1}(s), \nu_{1}(t)]$. In order to prove this, it suffices to consider the following three cases in view of the maximality of congruence $\theta$.

 \emph{Case I:} If $\alpha \in \overline{T}^{\theta}$, then we effectively need to show that $\overline{\alpha[s, t]}^{_{E_{\theta}}} = \overline{\alpha}^{_{\theta}} [ \overline{s}^{_{E_{\theta}}}, \overline{t}^{_{E_{\theta}}} ]$. From \rref{RemarkQuotient} and the fact that $\alpha \in \overline{T}^{_{\theta}}$, we have $\overline{\alpha}^{_{\theta}} [ \overline{s}^{_{E_{\theta}}}, \overline{t}^{_{E_{\theta}}} ] = \overline{s}^{_{E_{\theta}}}$. This reduces to show that $(\alpha[s, t], s) \in E_{\theta}$, which follows from \pref{PropAlphaST}(i).

 \emph{Case II:} In a similar vein, if $\alpha \in \overline{F}^{\theta}$, we need to show that $(\alpha[s, t], t) \in E_{\theta}$, which follows from \pref{PropAlphaST}(ii).

 \emph{Case III:} Similarly, if $\alpha \in \overline{U}^{\theta}$, we require that $(\alpha[s, t], \bot) \in E_{\theta}$, which is precisely \pref{PropAlphaST}(iii).

 This completes the proof.
\end{proof}

\begin{lemma} \label{LemmaEThetaM}
In case of the $C$-set $(M, M)$, the equivalence $E_{\theta}$ on $M$,  denoted by $E_{\theta_{M}}$, is a subset of $\theta$.
\end{lemma}

\begin{proof}
Let $(\alpha, \beta) \in E_{\theta_{M}}$. Then there exists $\gamma \in \overline{T}^{\theta}$ such that $\gamma \llbracket \alpha, \beta \rrbracket = \gamma \llbracket \beta, \beta \rrbracket$. In other words,
 \begin{equation} \label{EThetaM.1}
  (\gamma \wedge \alpha) \vee (\neg \gamma \wedge \beta) = (\gamma \wedge \beta) \vee (\neg \gamma \wedge \beta)
 \end{equation}
 Since $\gamma \in \overline{T}^{\theta}$, we have $(\gamma, T) \in \theta$. Moreover, $(\alpha, \alpha) \in \theta$ as $\theta$ is reflexive. It follows that $(\gamma \wedge \alpha, T \wedge \alpha) = (\gamma \wedge \alpha, \alpha) \in \theta$. Similarly, $(\gamma, T) \in \theta$ implies that $(\neg \gamma, F) \in \theta$, and as $(\beta, \beta) \in \theta$, we have $(\neg \gamma \wedge \beta, F \wedge \beta) = (\neg \gamma \wedge \beta, F) \in \theta$. Consequently $((\gamma  \wedge \alpha) \vee (\neg \gamma \wedge \beta), \alpha \vee F) = ((\gamma  \wedge \alpha) \vee (\neg \gamma \wedge \beta), \alpha) \in \theta$. Following a similar procedure, using $(\gamma, T), (\neg \gamma, F), (\beta, \beta) \in \theta$, we obtain $((\gamma  \wedge \beta) \vee (\neg \gamma \wedge \beta), \beta) \in \theta$. Now using \eqref{EThetaM.1}, the symmetry and transitivity of $\theta$, we have $(\alpha, \beta) \in \theta$ so that $E_{\theta_{M}} \subseteq \theta$.
\end{proof}

\begin{lemma} \label{LemmaIntersec}
$\displaystyle \bigcap_{\theta} E_{\theta} = \Delta_{S_{\bot}}$,  where $\theta$ ranges over all maximal congruences on $M$.
\end{lemma}

\begin{proof}
By \cref{CorBset}, $(S_{\bot}, M_{\#})$ is a $B$-set so that $(S_{\bot}, M_{\#})$ is a subdirect product of basic $B$-sets (cf. \tref{BsetSubdir}). Hence, $(S_{\bot}, M_{\#})$ is a subalgebra of a product of basic $B$-sets $(S_{x}, \mathbb{2})$, where $x$ ranges over some set $X$. That is,
 \begin{eqnarray*}
  (S_{\bot}, M_{\#}) & \leq& \prod_{x \in X} \left(S_{x}, \mathbb{2}\right) \\
                     & =& \left(\prod_{x \in X} S_{x}, \mathbb{2}^{X}\right) \\
                     & \leq& \left(\Big(\bigcup_{x \in X} S_{x}\Big)^{X}, \mathbb{2}^{X}\right) \\
 \end{eqnarray*}
 Note that the action in both $\left(\prod_{x \in X} S_{x}, \mathbb{2}^{X}\right)$ and $\left((\bigcup_{x \in X} S_{x})^{X}, \mathbb{2}^{X}\right)$ is
 \begin{equation*}
 (\alpha[s, t])(x) =
 \begin{cases}
  s(x), & \text{ if } \alpha(x) = T; \\
  t(x), & \text{ if } \alpha(x) = F.
 \end{cases}
 \end{equation*}
 Also note that the action in $\left(\prod_{x \in X} S_{x}, \mathbb{2}^{X}\right)$ is simply a restriction of that on $\left((\bigcup_{x \in X} S_{x})^{X}, \mathbb{2}^{X}\right)$. Since $(S_{\bot}, M_{\#})$ is a subalgebra of $\left((\bigcup_{x \in X} S_{x})^{X}, \mathbb{2}^{X}\right)$, we can see that $M_{\#}$ is a subalgebra of $\mathbb{2}^{X}$. Using the construction mentioned in \rref{RemarkStone}, $M \cong (M_{\#})^{\star} \leq \mathbb{3}^{X}$.

 Now for any $x_{o} \in X$, treating $M$ as a subalgebra of $\mathbb{3}^{X}$ we define maximal congruences on $M$ as follows.
 $$(\alpha, \beta) \in \theta_{x_{o}} \Leftrightarrow \alpha(x_{o}) = \beta(x_{o}).$$
 Such $\theta_{x_{o}}$ is indeed a maximal congruence on $M$. It is clearly an equivalence relation on $M$. Let $(\alpha_{1}, \beta_{1}), (\alpha_{2}, \beta_{2}) \in \theta_{x_{o}}$. Then $\alpha_{1}(x_{o}) = \beta_{1}(x_{o})$ and $\alpha_{2}(x_{o}) = \beta_{2}(x_{o})$. Thus $\alpha_{1}(x_{o}) \wedge \alpha_{2}(x_{o}) = \beta_{1}(x_{o}) \wedge \beta_{2}(x_{o})$. Thus $(\alpha_{1} \wedge \alpha_{2}, \beta_{1} \wedge \beta_{2}) \in \theta_{x_{o}}$. Similarly $\theta_{x_{o}}$ is compatible with the other operations on $M$, viz., $\vee, \neg$ and $^{\downarrow}$. Note that $M$ has only the following three equivalence classes with respect to $\theta_{x_{o}}$.
 \begin{align*}
 \overline{\bf T}^{\theta} & = \{ \alpha \in M : \alpha(x_{o}) = T \} \\
 \overline{\bf F}^{\theta} & = \{ \alpha \in M : \alpha(x_{o}) = F \} \\
 \overline{\bf U}^{\theta} & = \{ \alpha \in M : \alpha(x_{o}) = U \} \\
\end{align*}
 Thus $M / \theta_{x_{o}}$ is simple and so $\theta_{x_{o}}$ is maximal.

 We shall now show that $\bigcap E_{\theta} = \Delta_{S_{\bot}}$. Let $(s, t) \in \bigcap E_{\theta}$. Then for every maximal congruence $\theta$ on $M$, there exists a $\beta_{\theta} \in \overline{T}^{\theta}$ such that $\beta_{\theta}[s, t] = \beta_{\theta}[t, t]$. On using \pref{PropCsetarrow} we have $\beta_{\theta}^{\downarrow}[s, t] = \beta_{\theta}^{\downarrow}[t, t]$. Note that if $\beta_{\theta}$ is in $\overline{T}^{\theta}$, so is $\beta_{\theta}^{\downarrow}$. Moreover, $\beta_{\theta}^{\downarrow} \in M_{\#}$.

 As $S_{\bot} \leq (\bigcup_{x \in X} S_{x})^{X}$, we may treat $s, t$ as functions $s', t' \in (\bigcup_{x \in X} S_{x})^{X}$. Note that the {\tt if-then-else} action in $(S_{\bot}, M_{\#})$ can be treated as a restriction of that on $((\bigcup_{x \in X} S_{x})^{X}, \mathbb{2}^{X})$. Considering the maximal congruences defined above, for each $x_{o} \in X$ there exists $\beta_{\theta_{x_{o}}}^{\downarrow} \in \overline{T}^{\theta_{x_{o}}}$, that is, $\beta_{\theta_{x_{o}}}^{\downarrow}(x_{o}) = T$ and $\beta_{\theta_{x_{o}}}^{\downarrow}[s', t'] = \beta_{\theta_{x_{o}}}^{\downarrow}[t', t']$. In other words, for each $x \in X$, $(\beta_{\theta_{x_{o}}}^{\downarrow}[s', t'])(x) = (\beta_{\theta_{x_{o}}}^{\downarrow}[t', t'])(x)$. In particular for $x = x_{o}$, $(\beta_{\theta_{x_{o}}}^{\downarrow}[s', t'])(x_{o}) = (\beta_{\theta_{x_{o}}}^{\downarrow}[t', t'])(x_{o})$.

 However $(\beta_{\theta_{x_{o}}}^{\downarrow}[s', t'])(x_{o}) = s'(x_{o})$ as $\beta_{\theta_{x_{o}}}^{\downarrow}(x_{o}) = T$. Similarly, $(\beta_{\theta_{x_{o}}}^{\downarrow}[t', t'])(x_{o}) = t'(x_{o})$.

 This tells us that for each $x_{o} \in X$, $s'(x_{o}) = t'(x_{o})$, that is, $s' \equiv t'$ which means that $s = t$ in $S_{\bot}$. This completes the proof.
\end{proof}

\begin{remark}\label{Remarkinttheta}
For $\alpha , \beta \in M$ with $\alpha \ne \beta$, let $\theta_{\alpha, \beta}$ be a maximal congruence which separates $\alpha$ and $\beta$. Since $\displaystyle \bigcap_{\alpha \ne \beta \in M} \theta_{\alpha, \beta} = \Delta_{M}$, the intersection of all maximal congruences on $M$ \[\bigcap_{\theta \text{ maximal}} \theta = \Delta_{M}\]
\end{remark}

We shall now prove the main theorem of this section.

\begin{theorem} \label{CsetSubdir}
 Every $C$-set $(S_{\bot}, M)$ where $M$ is an ada is a subdirect product of basic $C$-sets.
\end{theorem}

\begin{proof}
 Let $(S_{\bot}, M)$ be a $C$-set where $M$ is an ada and $\{\theta\}$ be the collection of all maximal congruences on $M$. By \lref{LemmaCong}, for each $\theta$, the pair  $(E_{\theta}, \theta)$ is a $C$-set congruence on $(S_{\bot}, M)$ and by \rref{RemarkQuotient} $(S_{\bot} / E_{\theta}, M / \theta)$ is a basic $C$-set. Further, by \lref{LemmaIntersec} and \rref{Remarkinttheta}, the intersection of all congruences $(E_{\theta}, \theta)$ is trivial. Hence, $(S_{\bot}, M)$ is a subdirect product of $(S_{\bot} / E_{\theta}, \mathbb{3})$, where $\theta$ varies over maximal congruences on $M$.
\end{proof}

The following consequence of \tref{CsetSubdir} is useful to establish the equivalence between the programs which admit the current setup.

\begin{corollary} \label{IdCset}
 An identity $p \approx q$ is satisfied in every $C$-set $(S_{\bot}, M)$ where $M$ is an ada if and only if it is satisfied in all basic $C$-sets.
\end{corollary}

\section{Agreeable $C$-sets} \label{AgreeableCset}

In this section, we shall aim to describe an algebraic formalism for the equality test over possibly non-halting programs. The equality test over the functions $f, g \in \mathcal{T}_{o}(X_{\bot})$ can be naturally described by the following:

\begin{equation} \label{Agreeable}
(f * g)(x)  =
\begin{cases}
 T, & \text{ if } f(x) = g(x) \text{ and } f(x) \neq \bot \neq g(x); \\
 F, & \text{ if } f(x) \neq g(x) \text{ and } f(x) \neq \bot \neq g(x); \\
 U, & \text{ otherwise.}
\end{cases}
\end{equation}

For simplicity of notation, we will denote the condition $f(x) = g(x) \text{ and } f(x) \neq \bot \neq g(x)$ by $f(x) = g(x) \text{ } (\neq \bot)$ and the condition $f(x) \neq g(x) \text{ and } f(x) \neq \bot \neq g(x)$ by $f(x) \neq g(x) \text{ } (\neq \bot)$. Consequently, $f * g$ can be identified with the pair of sets $(A, B)$ on $X$, where $A = \{ x \in X : f(x) = g(x) \text{ } (\neq \bot) \}$ and $B = \{ x \in X : f(x) \neq g(x) \text{ } (\neq \bot) \}$.

Keeping this model in mind, we extend the notion of agreeable $B$-sets, given by Jackson and Stokes in \cite{jackson09}, and define agreeable $C$-sets as per the following.

\begin{definition}
 A $C$-set $(S_{\bot}, M)$ equipped with a function $$* : S_{\bot} \times S_{\bot} \rightarrow M$$ is said to be \emph{agreeable} if it satisfies the following axioms for all $s, t, u, v \in S_{\bot}$ and $\alpha \in M$:

 \begin{align}
  (s * s)[s, \bot] & = s \label{EA4} & \text{(domain axiom)} \\
  \bot * s  = U & = s * \bot \label{EA1} & \text{($\bot$-comparison)} \\
  (s * t)[s, t] & = (s * t)[t, t] \label{EA2} & \text{(equality on conclusions)} \\
  \alpha[s, t] * \alpha[u, v] & = \alpha \llbracket s * u, t * v \rrbracket \label{EA3} & \text{(operation interchange)} \\
  ((s * s = T) & \wedge (s * t = U)) \Rightarrow t = \bot \label{EA5} & \text{(totality condition)}
 \end{align}

\end{definition}

While the operation interchange axiom \eqref{EA3} is readily available, the axiom \eqref{EA2} can be verified in agreeable $B$-sets. However, the other axioms are specific to the current scenario of the non-halting case. These axioms can be justified along the following lines by considering equality of functions over the functional model  $(\mathcal{T}_{o}(X_{\bot}), \mathbb{3}^{X})$ of $C$-sets.

In $\mathcal{T}_{o}(X_{\bot})$, the domain of a function is considered in the spirit of a partial function, i.e., all those points whose image is not $\bot$. In the model $(\mathcal{T}_{o}(X_{\bot}), \mathbb{3}^{X})$, the partial predicate $s * s$ represents the domain of $s$. The domain axiom \eqref{EA4} captures the behaviour of {\tt if-then-else} with respect to the domain of $s$. For instance, we expect $s * s$ takes truth value $T$ in the domain of $s$ so that $(s * s)[s, \bot] = s$. Also, in the complement of the domain of $s$,  $s * s$ should take value $U$  so that $(s * s)[s, \bot] = U[s, \bot] = \bot = s$.

Note that we check the equality of two functions over their domains. Thus the $\bot$-comparison axiom  \eqref{EA1} states that comparing the {\tt error} state $\bot$ with any element $s$ results in the {\tt undefined} predicate $U$.

The axiom of equality on conclusions \eqref{EA2} exhibits the behaviour of equality test $*$ on conclusions of the {\tt if-then-else} action of the $C$-set. Indeed, when the partial predicate $s * t = T$, $(s * t)[s, t] = s = t = (s * t)[t, t]$ and similarly if $s * t = F$, then $(s * t)[s, t] = t = (s * t)[t, t]$. Further, if $s * t = U$, then  $(s * t)[s, t] = \bot = (s * t)[t, t]$.

The axiom of operation interchange \eqref{EA3} describes how $*$ and the {\tt if-then-else} action relate to the action on the $C$-set $(M, M)$. The totality condition \eqref{EA5} is a quasi-identity, in which if $s$ is a total function but $s * t$ is undefined, then it must follow that $t$ is the empty function, i.e., $t = \zeta_{\bot}$.

Thus we arrive at the following example of agreeable $C$-sets.

\begin{example} \label{ExampleFunctionalAgreeable}
 The pair $(\mathcal{T}_{o}(X_{\bot}), \mathbb{3}^{X})$ is an agreeable $C$-set under the operation $*$  defined in \eqref{Agreeable}. For verification of the axioms \eqref{EA4} -- \eqref{EA5} refer to Appendix \ref{VerFunctionalAgreeable}. Such agreeable $C$-sets are called \emph{agreeable functional $C$-sets}.
\end{example}

\begin{example} \label{ExampleBasicAgreeable}
 Every basic $C$-set is agreeable under the operation given by
 \begin{equation} \label{BasicAgreeableEq}
 s * t =
 \begin{cases}
  T, & \text{ if } s = t \text{ } (\neq \bot); \\
  F, & \text{ if } s \neq t \text{ } (\neq \bot); \\
  U, & \text{ if } s = \bot \text{ or } t = \bot.
 \end{cases}
 \end{equation}
 For verification of the axioms \eqref{EA4} -- \eqref{EA5} refer to Appendix \ref{VerBasicAgreeable}. Such agreeable $C$-sets will be called \emph{agreeable basic $C$-sets}.
 \end{example}

 \begin{proposition} \label{PropUniqueAgreeable}
  The operation defined in \eqref{BasicAgreeableEq} is the only possible operation under which a basic $C$-set can be made agreeable.
 \end{proposition}

 \begin{proof}
  We shall show that for a basic $C$-set, axioms \eqref{EA4} to \eqref{EA5} restrict the operation $*$ to precisely \eqref{BasicAgreeableEq}. Let $(S_{\bot}, \mathbb{3})$ be a basic $C$-set which is agreeable, that is, it is equipped with an operation $* : S_{\bot} \times S_{\bot} \rightarrow \mathbb{3}$ which satisfies \eqref{EA4} - \eqref{EA5}. Consider the following cases:
 \begin{enumerate}[(i)]
  \item \emph{Case I}: $s = \bot$ or $t = \bot$: Then from \eqref{EA1}, we have $s * t = U$.
  \item \emph{Case II}: $s = t \text{ } (\neq \bot)$: We will show that neither $s * t = F$ nor $s * t = U$ is possible. Consequently, it must be the case that $s * t = T$.

  Assume that $s * t = F$. This, in conjunction with the hypothesis $s = t$ and \eqref{EA4}, gives that $\bot = F[s, \bot] = (s * t)[s, \bot] = (s * s)[s, \bot] = s$, a contradiction to our assumption that $s \neq \bot$. If $s * t = U$, along similar lines, we obtain $\bot = U[s, \bot] = (s * s)[s, \bot] = s$, a contradiction.
  \item \emph{Case III}: $s \neq t \text{ } (\neq \bot)$: On similar lines, we will show that $s * t \notin \{ T, U \}$, which would imply that $s * t = F$.

  Assume that $s * t = T$. It follows from \eqref{EA2} that $s = T[s, t] = (s * t)[s, t] = (s * t)[t, t] = T[t, t] = t$, a contradiction to the hypothesis $s \neq t$. Note that if $S_{\bot}$ has exactly two distinct elements then this case would be redundant. Suppose that $s * t = U$. \emph{Case II} proved above, in conjunction with the hypothesis that $s \neq \bot$, give that $s * s = T$. From the statements $s * s = T$, $s * t = U$ and quasi-identity \eqref{EA5}, we have $t = \bot$, a contradiction.
 \end{enumerate}
 Thus the operation $*$ must be as defined in \eqref{BasicAgreeableEq}.
 \end{proof}

\begin{example} \label{ExampleCAlgAgreeable}
 The $C$-set $(M, M)$ is agreeable under the operation
 $$\alpha * \beta = (\alpha \wedge \beta) \vee (\neg \alpha \wedge \neg \beta).$$
 The operation can be equivalently expressed in terms of the {\tt if-then-else} action by
 $$\alpha * \beta = \alpha \llbracket \beta, \neg \beta \rrbracket.$$
 For verification of the axioms \eqref{EA4} -- \eqref{EA5} refer to Appendix \ref{VerCAlgAgreeable}.
\end{example}

\begin{remark}
 If the $C$-algebra $M$ is  $\mathbb{3}^{X}$, the equality test on the agreeable $C$-set $(M, M)$ coincides with that of the functional case, as shown below:
 \begin{equation*}
  (\alpha * \beta)(x) =
  \begin{cases}
   T, & \text{ if } \alpha(x) = \beta(x) \text{ } (\neq U); \\
   F, & \text{ if } \alpha(x) \neq \beta(x) \text{ } (\neq U); \\
   U, & \text{ otherwise.}
  \end{cases}
 \end{equation*}
\end{remark}

We shall prove a representation theorem of agreeable $C$-sets along the lines of \tref{CsetSubdir}.

\begin{theorem} \label{TheoremRepAgreeableCSet}
Every agreeable $C$-set where $M$ is an ada is a subdirect product of agreeable basic $C$-sets.
\end{theorem}

\begin{proof}
 Let $(S_{\bot}, M)$ be an agreeable $C$-set where $M$ is an ada. For every maximal congruences  $\theta$ on $M$, consider the pair $(E_{\theta}, \theta)$ as in \dref{DefinitionETheta}. By \lref{LemmaCong}, we have already ascertained that for each $\theta$, the pair  $(E_{\theta}, \theta)$ is a $C$-set congruence on $(S_{\bot}, M)$ and by \rref{RemarkQuotient} that $(S_{\bot} / E_{\theta}, M / \theta)$ is a basic $C$-set. In order to ascertain that this pair is indeed a congruence in the context of agreeable $C$-sets, it is sufficient to show that
 $$(a_{1}, a_{2}), (b_{1}, b_{2}) \in E_{\theta} \Rightarrow (a_{1} * b_{1}, a_{2} * b_{2}) \in \theta.$$
 Let $(a_{1}, a_{2}), (b_{1}, b_{2}) \in E_{\theta}$. Then there exist $\alpha \text{ and } \beta \in \overline{T}^{\theta}$ such that
 \begin{align}
  \alpha[a_{1}, a_{2}] & = \alpha[a_{2}, a_{2}] \label{eq1} \\
  \beta[b_{1}, b_{2}] & = \beta[b_{2}, b_{2}] \label{eq2}
 \end{align}
 Note that $(\alpha, T), (\beta, T) \in \theta$ implies that $(\alpha \wedge \beta, T \wedge T) = (\alpha \wedge \beta, T) \in \theta$. Applying \eqref{EC8} on \eqref{eq1} and \pref{PropCset}(v) on \eqref{eq2}, we have, for $(\alpha \wedge \beta) \in \overline{T}^{\theta}$ \\
 \begin{align}
 (\alpha \wedge \beta)[a_{1}, a_{2}] & = (\alpha \wedge \beta)[a_{2}, a_{2}] \label{eq3} \\
 (\alpha \wedge \beta)[b_{1}, b_{2}] & = (\alpha \wedge \beta)[b_{2}, b_{2}] \label{eq4}
 \end{align}
 These imply that
 \begin{equation}
  (\alpha \wedge \beta)[a_{1}, a_{2}] * (\alpha \wedge \beta)[b_{1}, b_{2}] = (\alpha \wedge \beta)[a_{2}, a_{2}] * (\alpha \wedge \beta)[b_{2}, b_{2}].
  \end{equation}
 From \eqref{EA3} it follows that
 $$(\alpha \wedge \beta) \llbracket a_{1} * b_{1}, a_{2} * b_{2} \rrbracket = (\alpha \wedge \beta) \llbracket a_{2} * b_{2}, a_{2} * b_{2} \rrbracket,$$
 so that $(a_{1} * b_{1}, a_{2} * b_{2}) \in E_{\theta_{M}} \subseteq \theta$, by \lref{LemmaEThetaM}. Further, by \lref{LemmaIntersec} and \rref{Remarkinttheta}, the intersection of all congruences $(E_{\theta}, \theta)$,  where $\theta$ ranges over all maximal congruences of $M$, is trivial. This completes the proof.
\end{proof}

\begin{corollary} \label{CorollaryIdentityAgreeable}
 An identity $p \approx q$ is satisfied in every agreeable $C$-set $(S_{\bot}, M)$ where $M$ is an ada if and only if it is satisfied in all agreeable basic $C$-sets.
\end{corollary}

In view of \cref{CorollaryIdentityAgreeable} and \eqref{BasicAgreeableEq}, we have the following result.

\begin{corollary}
 In every agreeable $C$-set where $M$ is an ada we have $s * t = t * s$.
\end{corollary}

Note that the only axiom of agreeable $C$-sets that plays a role in the proof of \tref{TheoremRepAgreeableCSet} is \eqref{EA3}. The remaining axioms have been included in order that the operation on agreeable basic $C$-sets be uniquely defined. The proof of \tref{TheoremRepAgreeableCSet} suggests an alternative proof for \tref{AgreeableBsetSubdir} of \cite{jackson09}, without using the commutativity of $*$.  We provide an alternative proof in the following.

\begin{theorem}[\cite{jackson09}] \label{Prop.JS.Agreeable}
 Every agreeable $B$-set $(S, B)$ is a subdirect product of basic $B$-sets.
\end{theorem}

\begin{proof}
 Let $F$ be an ultrafilter of $B$. Consider the relation $E_{F} = \{ (s, t) \in S \times S : \gamma[s, t] = t \text{ for some } \gamma \in F \}$ as defined in \cite{jackson09}. The pair $(E_{F}, F)$ is a $B$-set congruence. In order to show that the pair $(E_{F}, F)$ is a congruence on agreeable $B$-sets, we show that
 $$(a_{1}, a_{2}), (b_{1}, b_{2}) \in E_{F} \Rightarrow (a_{1} * b_{1}, a_{2} * b_{2}) \in \theta_{F},$$ where $\theta_{F}$ is the congruence on $B$ induced by the ultrafilter $F$.

 Since $(a_{1}, a_{2}), (b_{1}, b_{2}) \in E_{F}$, there exist $\alpha, \beta \in F$ such that
 \begin{align}
  \alpha[a_{1}, a_{2}] & = a_{2} \label{eq3} \\
  \beta[b_{1}, b_{2}] & = b_{2} \label{eq4}
 \end{align}
 In view of the commutativity of $\wedge$, \eqref{B6}, \eqref{eq3} and \eqref{B1}, we obtain $(\alpha \wedge \beta)[a_{1}, a_{2}] = (\beta \wedge \alpha)[a_{1}, a_{2}] = \beta[\alpha[a_{1}, a_{2}], a_{2}] = \beta[a_{2}, a_{2}] = a_{2}$. Similarly we can obtain \break $(\alpha \wedge \beta)[b_{1}, b_{2}] = b_{2}$. This implies that
 \begin{equation*}
  (\alpha \wedge \beta)[a_{1}, a_{2}] * (\alpha \wedge \beta)[b_{1}, b_{2}] = a_{2} * b_{2}
 \end{equation*}
 From axiom \eqref{AB3}, we can deduce that $$(\alpha \wedge \beta)[a_{1} * b_{1}, a_{2} * b_{2}] = a_{2} * b_{2}.$$
 Since $F$ is an ultrafilter of $B$, it suffices to ascertain that
 $$a_{1} * b_{1} \in F \Leftrightarrow a_{2} * b_{2} \in F.$$

 Assume that $a_{1} * b_{1} \in F$. Since $\alpha \wedge \beta \in F$, we have $(\alpha \wedge \beta) \wedge (a_{1} * b_{1}) \in F$. Further, as $F$ is a filter and $(\alpha \wedge \beta) \wedge (a_{1} * b_{1}) \leq ((\alpha \wedge \beta) \wedge (a_{1} * b_{1})) \vee (\neg (\alpha \wedge \beta) \wedge (a_{2} * b_{2})) = (\alpha \wedge \beta) [ a_{1} * b_{1}, a_{2} * b_{2} ]$ we have $(\alpha \wedge \beta) [ a_{1} * b_{1}, a_{2} * b_{2} ] = a_{2} * b_{2} \in F$.

 Conversely, assume that $a_{2} * b_{2} \in F$. The symmetry of equivalence relation $E_{F}$ implies that $(a_{2}, a_{1}), (b_{2}, b_{1}) \in E_{F}$. Along similar lines as above, we obtain $a_{1} * b_{1} \in F$. This completes the proof.
\end{proof}

\section{Conclusion}

The axiomatization of systems based on various program constructs are extremely useful in the study of program semantics, in general, and in establishing program equivalence, in particular. While many authors have studied the axiomatization of the {\tt if-then-else} construct, this work considered the case where the programs and tests could possibly be non-halting. In this connection, this work introduced the notion of $C$-sets to axiomatize the systems of {\tt if-then-else} in which the tests are drawn from an abstract $C$-algebra. The axioms of $C$-sets include a quasi-identity for $\wedge$-compatibility along with various other identities.  When the $C$-algebra is an ada, we obtained a subdirect representation of $C$-sets through basic $C$-sets. This in turn establishes the completeness of the axiomatization and paves the way for determining the equivalence of programs under consideration through basic $C$-sets.
Further, in order to axiomatize {\tt if-then-else} systems with the equality test, this work extended the concept of $C$-sets to agreeable $C$-sets and obtained similar results.

To study further, one may investigate a complete axiomatization of the systems under consideration in which all the axioms are equations. It is also desirable to extend the results to the general case of $C$-sets without restricting the $C$-algebra to be an ada.

\appendix

\section{Proofs} \label{SectionAnnexure}

\subsection{Verification of \eref{e-m-m}} \label{VerCAlgCset}

\begin{flushleft}

 Axiom \eqref{EC1}: We will show that $U \llbracket \alpha, \beta \rrbracket = U$. Here $U \llbracket \alpha, \beta \rrbracket = (U \wedge \alpha) \vee (\neg U \wedge \beta) = (U \wedge \alpha) \vee (U \wedge \beta) = U \vee U = U$. \newline

 Axiom \eqref{EC6}: By definition we have $F \llbracket \alpha, \beta \rrbracket = (F \wedge \alpha) \vee (\neg F \wedge \beta) = (F \wedge \alpha) \vee (T \wedge \beta) = F \vee \beta = \beta$. \newline

 Axiom \eqref{EC5}: We have $(\neg \alpha) \llbracket \beta, \gamma \rrbracket = (\neg \alpha \wedge \beta) \vee ( \alpha \wedge \gamma)$ while $\alpha \llbracket \gamma, \beta \rrbracket = (\alpha \wedge \gamma) \vee (\neg \alpha \wedge \beta)$. Thus we have to check the validity of the identity $(\neg \alpha \wedge \beta) \vee ( \alpha \wedge \gamma) = (\alpha \wedge \gamma) \vee (\neg \alpha \wedge \beta)$. In view of \tref{SubdirIrredCAlg}, this identity is valid in all $C$-algebras if it is valid in the three element $C$-algebra $\mathbb{3}$.
 \begin{enumerate}[\rm]
  \item \emph{Case I}: $\alpha = T$: In this case $\neg \alpha \llbracket \beta, \gamma \rrbracket = (\neg T \wedge \beta) \vee (T \wedge \gamma) = (F \wedge \beta) \vee (T \wedge \gamma) = F \vee \gamma = \gamma$. On the other hand $\alpha \llbracket \gamma, \beta \rrbracket = (T \wedge \gamma) \vee (\neg T \wedge \beta) = \gamma \vee F = \gamma$.
  \item \emph{Case II}: $\alpha = F$: Here we have $\neg \alpha \llbracket \beta, \gamma \rrbracket = (\neg F \wedge \beta) \vee (F \wedge \gamma) = \beta \vee F = \beta$. Also $\alpha \llbracket \gamma, \beta \rrbracket = (F \wedge \gamma) \vee (\neg F \wedge \beta) = F \vee \beta = \beta$.
  \item \emph{Case III}: $\alpha = U$: This reduces to the verification of axiom \eqref{EC1}, from which it follows that both expressions evaluate to $U$. \newline
 \end{enumerate}

 Axiom \eqref{EC3}: We will show that $\alpha \llbracket \alpha \llbracket \beta, \gamma \rrbracket, \delta \rrbracket = \alpha \llbracket \beta, \delta \rrbracket$. Here $\alpha \llbracket \alpha \llbracket \beta, \gamma \rrbracket, \delta \rrbracket = (\alpha \wedge ((\alpha \wedge \beta) \vee (\neg \alpha \wedge \gamma))) \vee (\neg \alpha \wedge \delta)$. Also $\alpha \llbracket \beta, \delta \rrbracket = (\alpha \wedge \beta) \vee (\neg \alpha \wedge \delta)$. As above, the verification reduces to checking the validity in $C$-algebra $\mathbb{3}$.
 \begin{enumerate}[\rm]
  \item \emph{Case I}: $\alpha = T$: In this case $\alpha \llbracket \alpha \llbracket \beta, \gamma \rrbracket, \delta \rrbracket = (T \wedge ((T \wedge \beta) \vee (\neg T \wedge \gamma))) \vee (\neg T \wedge \delta) = (\beta \vee F) \vee F = \beta$. Also $\alpha \llbracket \beta, \delta \rrbracket = (T \wedge \beta) \vee (\neg T \wedge \delta) = \beta \vee F = \beta$.
  \item \emph{Case II}: $\alpha = F$: Here we have $\alpha \llbracket \alpha \llbracket \beta, \gamma \rrbracket, \delta \rrbracket = (F \wedge ((F \wedge \beta) \vee (\neg F \wedge \gamma))) \vee (\neg F \wedge \delta) = F \vee \delta = \delta$. Also $\alpha \llbracket \beta, \delta \rrbracket = (F \wedge \beta) \vee (\neg F \wedge \delta) = F \vee \delta = \delta$.
  \item \emph{Case III}: $\alpha = U$: Since $\neg U = U$ and $U$ is a left-identity for $\wedge$ and $\vee$ in this case, both expressions evaluate to $U$. \newline
 \end{enumerate}

 Axiom \eqref{EC4}: Here $\alpha \llbracket \beta, \alpha \llbracket \gamma, \delta \rrbracket \rrbracket = (\alpha \wedge \beta) \vee (\neg \alpha \wedge ((\alpha \wedge \gamma) \vee (\neg \alpha \wedge \delta)))$ while $\alpha \llbracket \beta, \delta \rrbracket = (\alpha \wedge \beta) \vee (\neg \alpha \wedge \delta)$. In view of \tref{SubdirIrredCAlg} it suffices to consider the following three cases:
 \begin{enumerate}[\rm]
  \item \emph{Case I}: $\alpha = T$: Here $(\alpha \wedge \beta) \vee (\neg \alpha \wedge ((\alpha \wedge \gamma) \vee (\neg \alpha \wedge \delta))) = (T \wedge \beta) \vee (\neg T \wedge ((T \wedge \gamma) \vee (\neg T \wedge \delta))) = \beta \vee F = \beta$ while $(\alpha \wedge \beta) \vee (\neg \alpha \wedge \delta) = (T \wedge \beta) \vee (\neg T \wedge \delta) = \beta \vee F = \beta$.
  \item \emph{Case II}: $\alpha = F$: In this case $(\alpha \wedge \beta) \vee (\neg \alpha \wedge ((\alpha \wedge \gamma) \vee (\neg \alpha \wedge \delta))) = (F \wedge \beta) \vee (\neg F \wedge ((F \wedge \gamma) \vee (\neg F \wedge \delta))) = F \vee (F \vee \delta) = \delta$. On the other hand $(\alpha \wedge \beta) \vee (\neg \alpha \wedge \delta) = (F \wedge \beta) \vee (\neg F \wedge \delta) = F \vee \delta = \delta$.
  \item \emph{Case III}: $\alpha = U$: It is easy to see that both expressions in this case evaluate to $U$. \newline
 \end{enumerate}

 Axiom \eqref{EC7}: Here $(\alpha \wedge \beta) \llbracket \gamma, \delta \rrbracket = ((\alpha \wedge \beta) \wedge \gamma) \vee (\neg (\alpha \wedge \beta) \wedge \delta)$ while $\alpha \llbracket \beta \llbracket \gamma, \delta \rrbracket, \delta \rrbracket = (\alpha \wedge ((\beta \wedge \gamma) \vee (\neg \beta \wedge \delta))) \vee (\neg \alpha \wedge \delta)$. It suffices to consider the following three cases:
 \begin{enumerate}[\rm]
  \item \emph{Case I}: $\alpha = T$: Here $(\alpha \wedge \beta) \llbracket \gamma, \delta \rrbracket = ((T \wedge \beta) \wedge \gamma) \vee (\neg (T \wedge \beta) \wedge \delta) = (\beta \wedge \gamma) \vee (\neg \beta \wedge \delta)$. On the other hand $\alpha \llbracket \beta \llbracket \gamma, \delta \rrbracket, \delta \rrbracket = (T \wedge ((\beta \wedge \gamma) \vee (\neg \beta \wedge \delta))) \vee (\neg T \wedge \delta) = ((\beta \wedge \gamma) \vee (\neg \beta \wedge \delta)) \vee F = (\beta \wedge \gamma) \vee (\neg \beta \wedge \delta)$.
  \item \emph{Case II}: $\alpha = F$: In this case we have $(\alpha \wedge \beta) \llbracket \gamma, \delta \rrbracket = ((F \wedge \beta) \wedge \gamma) \vee (\neg (F \wedge \beta) \wedge \delta) = F \vee \delta = \delta$. Also $\alpha \llbracket \beta \llbracket \gamma, \delta \rrbracket, \delta \rrbracket = (F \wedge ((\beta \wedge \gamma) \vee (\neg \beta \wedge \delta))) \vee (\neg F \wedge \delta) = F \vee \delta = \delta$.
  \item \emph{Case III}: $\alpha = U$: It is easy to see that both expressions in this case evaluate to $U$. \newline
 \end{enumerate}

 Axiom \eqref{EC2}: Here $\alpha \llbracket \beta \llbracket \gamma, \delta \rrbracket, \beta \llbracket \rho, \omega \rrbracket \rrbracket = (\alpha \wedge ((\beta \wedge \gamma) \vee (\neg \beta \wedge \delta ))) \vee (\neg \alpha \wedge ((\beta \wedge \rho) \vee (\neg \beta \wedge \omega)))$. On the other hand $\beta \llbracket \alpha \llbracket \gamma, \rho \rrbracket, \alpha \llbracket \delta, \omega \rrbracket \rrbracket = (\beta \wedge ((\alpha \wedge \gamma) \vee (\neg \alpha \wedge \rho ))) \vee (\neg \beta \wedge ((\alpha \wedge \delta) \vee (\neg \alpha \wedge \omega)))$. It suffices to consider the following three cases:
 \begin{enumerate}[\rm]
  \item \emph{Case I}: $\alpha = T$: Here $\alpha \llbracket \beta \llbracket \gamma, \delta \rrbracket, \beta \llbracket \rho, \omega \rrbracket \rrbracket = (T \wedge ((\beta \wedge \gamma) \vee (\neg \beta \wedge \delta ))) \vee (\neg T \wedge ((\beta \wedge \rho) \vee (\neg \beta \wedge \omega))) = ((\beta \wedge \gamma) \vee (\neg \beta \wedge \delta )) \vee F = (\beta \wedge \gamma) \vee (\neg \beta \wedge \delta)$. Also $\beta \llbracket \alpha \llbracket \gamma, \rho \rrbracket, \alpha \llbracket \delta, \omega \rrbracket \rrbracket = (\beta \wedge ((T \wedge \gamma) \vee (\neg T \wedge \rho ))) \vee (\neg \beta \wedge ((T \wedge \delta) \vee (\neg T \wedge \omega))) = (\beta \wedge (\gamma \vee F)) \vee (\neg \beta \wedge (\delta \vee F)) = (\beta \wedge \gamma) \vee (\neg \beta \wedge \delta)$.
  \item \emph{Case II}: $\alpha = F$: In this case $\alpha \llbracket \beta \llbracket \gamma, \delta \rrbracket, \beta \llbracket \rho, \omega \rrbracket \rrbracket = (F \wedge ((\beta \wedge \gamma) \vee (\neg \beta \wedge \delta ))) \vee (\neg F \wedge ((\beta \wedge \rho) \vee (\neg \beta \wedge \omega))) = F \vee ((\beta \wedge \rho) \vee (\neg \beta \wedge \omega)) = (\beta \wedge \rho) \vee (\neg \beta \wedge \omega)$. Similarly $\beta \llbracket \alpha \llbracket \gamma, \rho \rrbracket, \alpha \llbracket \delta, \omega \rrbracket \rrbracket = (\beta \wedge ((F \wedge \gamma) \vee (\neg F \wedge \rho ))) \vee (\neg \beta \wedge ((F \wedge \delta) \vee (\neg F \wedge \omega))) = (\beta \wedge (F \vee \rho)) \vee (\neg \beta \wedge (F \vee \omega)) = (\beta \wedge \rho) \vee (\neg \beta \wedge \omega)$.
  \item \emph{Case III}: $\alpha = U$: In this case both expressions evaluate to $U$ by checking casewise for $\beta \in \{ T, F, U \}$. \newline
 \end{enumerate}

 Axiom \eqref{EC8}: In view of \tref{SubdirIrredCAlg} we have $M \leq \mathbb{3}^{X}$ for some set $X$. Since $\alpha \llbracket \gamma, \delta \rrbracket = \alpha \llbracket \delta, \delta \rrbracket$ we have $(\alpha \llbracket \gamma, \delta \rrbracket)(x) = (\alpha \llbracket \delta, \delta \rrbracket)(x)$ for all $x \in X$ by treating $\alpha, \gamma, \delta$ as elements of $\mathbb{3}^{X}$. This reduces to the expression $(\alpha(x) \wedge \gamma(x)) \vee (\neg \alpha(x) \wedge \delta(x)) = (\alpha(x) \wedge \delta(x)) \vee (\neg \alpha(x) \wedge \delta(x))$ which holds for all $x \in X$. Consider $\beta \in M$ as an element in $\mathbb{3}^{X}$. Now $((\alpha \wedge \beta) \llbracket \gamma, \delta \rrbracket)(x) = ((\alpha(x) \wedge \beta(x)) \wedge \gamma(x)) \vee (\neg (\alpha(x) \wedge \beta(x)) \wedge \delta(x))$. On similar lines $((\alpha \wedge \beta) \llbracket \delta, \delta \rrbracket)(x) = ((\alpha(x) \wedge \beta(x)) \wedge \delta(x)) \vee (\neg (\alpha(x) \wedge \beta(x)) \wedge \delta(x))$. It suffices to consider the following three cases:
 \begin{enumerate}[\rm]
  \item \emph{Case I}: $\alpha(x) = T$: From the hypothesis we have $(\alpha(x) \wedge \gamma(x)) \vee (\neg \alpha(x) \wedge \delta(x)) = (\alpha(x) \wedge \delta(x)) \vee (\neg \alpha(x) \wedge \delta(x))$ that is $(T \wedge \gamma(x)) \vee (\neg T \wedge \delta(x)) = (T \wedge \delta(x)) \vee (\neg T \wedge \delta(x))$. Consequently we have $\gamma(x) = \delta(x)$. Thus $((\alpha \wedge \beta) \llbracket \gamma, \delta \rrbracket)(x) = ((T \wedge \beta(x)) \wedge \gamma(x)) \vee (\neg (T \wedge \beta(x)) \wedge \delta(x)) = (\beta(x) \wedge \gamma(x)) \vee (\neg \beta(x) \wedge \delta(x)) = (\beta(x) \wedge \delta(x)) \vee (\neg \beta(x) \wedge \delta(x)) = ((T \wedge \beta(x)) \wedge \delta(x)) \vee (\neg (T \wedge \beta(x)) \wedge \delta(x)) = ((\alpha \wedge \beta) \llbracket \delta, \delta \rrbracket)(x)$.
  \item \emph{Case II}: $\alpha(x) = F$: In this case $((\alpha \wedge \beta) \llbracket \gamma, \delta \rrbracket)(x) = ((F \wedge \beta(x)) \wedge \gamma(x)) \vee (\neg (F \wedge \beta(x)) \wedge \delta(x)) = F \vee \delta(x) = ((F \wedge \beta(x)) \wedge \delta(x)) \vee (\neg (F \wedge \beta(x)) \wedge \delta(x)) = ((\alpha \wedge \beta) \llbracket \delta, \delta \rrbracket)(x)$.
  \item \emph{Case III}: $\alpha(x) = U$: It is easy to see that $((\alpha \wedge \beta) \llbracket \gamma, \delta \rrbracket)(x) = U = ((\alpha \wedge \beta) \llbracket \delta, \delta \rrbracket)(x)$.
 \end{enumerate}
 Thus $((\alpha \wedge \beta) \llbracket \gamma, \delta \rrbracket)(x) = ((\alpha \wedge \beta) \llbracket \delta, \delta \rrbracket)(x)$ for all $x \in X$. This implies that the identity $(\alpha \wedge \beta) \llbracket \gamma, \delta \rrbracket = (\alpha \wedge \beta) \llbracket \delta, \delta \rrbracket$ holds in the $C$-algebra $M$. \newline

\end{flushleft}

\subsection{Verification of \eref{ExampleFunctionalCset}} \label{VerFunctionalCset}

\begin{flushleft}

In order to verify the axioms we will rely on the pairs of sets representation of the $C$-algebra $\mathbb{3}^{X}$ by Guzm\'{a}n and Squier in \cite{guzman90}. Every $\alpha \in \mathbb{3}^{X}$ can be represented by the pair of sets $(A, B) = (\alpha^{-1}(T), \alpha^{-1}(F))$. In this representation ${\bf T} = (X, \emptyset), {\bf F} = (\emptyset, X)$ and ${\bf U} = (\emptyset, \emptyset)$. Thus

\begin{equation*}
 \alpha[f, g](x) = (A, B)[f, g](x) = \begin{cases}
                                      f(x), & \text{ if } x \in A; \\
                                      g(x), & \text{ if } x \in B; \\
                                      \bot, & \text{ otherwise.}
                                     \end{cases}
\end{equation*}

 Axiom \eqref{EC1}: In view of this notation we have ${\bf U}[f, g](x) = (\emptyset, \emptyset)[f, g](x) = \bot$ for all $x \in X_{\bot}$. Thus ${\bf U}[f, g] = \zeta_{\bot}$. \newline

 Axiom \eqref{EC6}: We have ${\bf F}[f, g](x) = (\emptyset, X)[f, g](x) = g(x)$ for all $x \in X_{\bot}$. It follows that ${\bf F}[f, g] = g$. \newline

 Axiom \eqref{EC5}: If $\alpha \in \mathbb{3}^{X}$ is represented by $(A, B)$ then $\neg \alpha$ is represented by $(B, A)$. Thus $(\neg \alpha)[f, g](x) = (B, A)[f, g](x)$. It follows that

 \begin{align*}
  (B, A)[f, g](x) & =  \begin{cases}
                         f(x), & \text{ if } x \in B; \\
                         g(x), & \text{ if } x \in A; \\
                         \bot, & \text{ otherwise.} \\
                        \end{cases} \\
                  & =  \begin{cases}
                         g(x), & \text{ if } x \in A; \\
                         f(x), & \text{ if } x \in B; \\
                         \bot, & \text{ otherwise.} \\
                        \end{cases} \\
                  & =   (A, B)[g, f](x).
 \end{align*}

 Thus $(\neg \alpha)[f, g] = \alpha[g, f]$. \newline

 Axiom \eqref{EC3}: Let $\alpha \in \mathbb{3}^{X}$ be represented by $(A, B)$.
 \begin{equation*}
 (A, B)[(A, B)[f, g], h](x) = \begin{cases}
                               (A, B)[f, g](x), & \text{ if } x \in A; \\
                               h(x), & \text{ if } x \in B; \\
                               \bot, & \text{ otherwise.}
                              \end{cases}
 \end{equation*}

 Since $(A, B)[f, g](x) = f(x)$ when $x \in A$ we have
 \begin{align*}
 (A, B)[(A, B)[f, g], h](x) & = \begin{cases}
                                f(x), & \text{ if } x \in A; \\
                                h(x), & \text{ if } x \in B; \\
                                \bot, & \text{ otherwise.} \\
                               \end{cases} \\
                            & = (A, B)[f, h](x).
 \end{align*}

 Thus $\alpha[\alpha[f, g],  h] = \alpha[f, h]$. \newline

 Axiom \eqref{EC4}: Let $(A, B)$ represent $\alpha \in \mathbb{3}^{X}$.

 \begin{equation*}
 (A, B)[f, (A, B)[g, h]](x) = \begin{cases}
                               f(x), & \text{ if } x \in A; \\
                               (A, B)[g, h](x), & \text{ if } x \in B; \\
                               \bot, & \text{ otherwise.} \newline
                              \end{cases}
 \end{equation*}

 Since $(A, B)[g, h](x) = h(x)$ when $x \in B$ we have

 \begin{align*}
 (A, B)[f, (A, B)[g, h]](x) & = \begin{cases}
                                f(x), & \text{ if } x \in A; \\
                                h(x), & \text{ if } x \in B; \\
                                \bot, & \text{ otherwise.} \\
                               \end{cases} \\
                            & = (A, B)[f, h](x). \end{align*}

 Thus $\alpha[f, \alpha[g, h]] = \alpha[f, h]$. \newline

 Axiom \eqref{EC7}: Let $\alpha$ and $\beta \in \mathbb{3}^{X}$ be represented by $(A, B)$ and $(C, D)$ respectively. We have $(A, B) \wedge (C, D) = (A \cap C, B \cup (A \cap D))$. Thus we have the following:

 \begin{equation*}
 (\alpha \wedge \beta)[f, g](x) = (A \cap C, B \cup (A \cap D))[f, g](x) = \begin{cases}
                                                                            f(x), & \text{ if } x \in A \cap C; \\
                                                                            g(x), & \text{ if } x \in B \cup (A \cap D); \\
                                                                            \bot, & \text{ otherwise.}
                                                                           \end{cases}
 \end{equation*}

 Similarly we have

 \begin{align*}
 (A, B)[(C, D)[f, g], g](x) & = \begin{cases}
                                (C, D)[f, g](x), & \text{ if } x \in A; \\
                                g(x), & \text{ if } x \in B; \\
                                \bot, & \text{ otherwise.} \\
                               \end{cases} \\
                           & = \begin{cases}
                                f(x), & \text{ if } x \in A \cap C; \\
                                g(x), & \text{ if } x \in A \cap D; \\
                                \bot, & \text{ if } x \in A \cap (X \setminus (C \cup D)); \\
                                g(x), & \text{ if } x \in B; \\
                                \bot, & \text{ otherwise.} \\
                               \end{cases} \\
                           & = \begin{cases}
                                f(x), & \text{ if } x \in A \cap C; \\
                                g(x), & \text{ if } x \in B \cup (A \cap D); \\
                                \bot, & \text{ otherwise.}
                               \end{cases}
 \end{align*}

 Thus $(A \cap C, B \cup (A \cap D))[f, g](x) = (A, B)[(C, D)[f, g], g](x)$ for all $x \in X_{\bot}$ and so $(\alpha \wedge \beta)[f, g] = \alpha[\beta[f, g], g]$. \newline

 Axiom \eqref{EC2}: As earlier let $\alpha$ and $\beta \in \mathbb{3}^{X}$ be represented by $(A, B)$ and $(C, D)$ respectively.

 \begin{align*}
  (A, B)[(C, D)[f, g], (C, D)[h, k]](x) & = \begin{cases}
                                            (C, D)[f, g](x), & \text{ if } x \in A; \\
                                            (C, D)[h, k](x), & \text{ if } x \in B; \\
                                            \bot, & \text{ otherwise.} \\
                                           \end{cases} \\
                                        & = \begin{cases}
                                            f(x), & \text{ if } x \in A \cap C; \\
                                            g(x), & \text{ if } x \in A \cap D; \\
                                            \bot, & \text{ if } x \in A \cap (X \setminus (C \cup D)); \\
                                            h(x), & \text{ if } x \in B \cap C; \\
                                            k(x), & \text{ if } x \in B \cap D; \\
                                            \bot, & \text{ if } x \in B \cap (X \setminus (C \cup D)); \\
                                            \bot, & \text{ otherwise.} \\
                                           \end{cases} \\
                                       & = \begin{cases}
                                            f(x), & \text{ if } x \in A \cap C; \\
                                            g(x), & \text{ if } x \in A \cap D; \\
                                            h(x), & \text{ if } x \in B \cap C; \\
                                            k(x), & \text{ if } x \in B \cap D; \\
                                            \bot, & \text{ otherwise.}
                                           \end{cases}
 \end{align*}

 Also we have

 \begin{align*}
  (C, D)[(A, B)[f, h], (A, B)[g, k]](x) & = \begin{cases}
                                             (A, B)[f, h](x), & \text{ if } x \in C; \\
                                             (A, B)[g, k](x), & \text{ if } x \in D; \\
                                             \bot, & \text{ otherwise.} \\
                                            \end{cases} \\
                                        & = \begin{cases}
                                            f(x), & \text{ if } x \in C \cap A; \\
                                            h(x), & \text{ if } x \in C \cap B; \\
                                            \bot, & \text{ if } x \in C \cap (X \setminus (A \cup B)); \\
                                            g(x), & \text{ if } x \in D \cap A; \\
                                            k(x), & \text{ if } x \in D \cap B; \\
                                            \bot, & \text{ if } x \in D \cap (X \setminus (A \cup B)); \\
                                            \bot, & \text{ otherwise.} \\
                                           \end{cases} \\
                                       & = \begin{cases}
                                            f(x), & \text{ if } x \in C \cap A; \\
                                            h(x), & \text{ if } x \in C \cap B; \\
                                            g(x), & \text{ if } x \in D \cap A; \\
                                            k(x), & \text{ if } x \in D \cap B; \\
                                            \bot, & \text{ otherwise.} \\
                                           \end{cases} \\
 \end{align*}

 Thus $(A, B)[(C, D)[f, g], (C, D)[h, k]](x) = (C, D)[(A, B)[f, h], (A, B)[g, k]](x)$ for all $x \in X_{\bot}$ and so $\alpha[\beta[f, g], \beta[h, k]] = \beta[\alpha[f, h], \alpha[g, k]]$. \newline

 Axiom \eqref{EC8}: Let $\alpha, \beta \in \mathbb{3}^{X}$ be represented by $(A, B)$ and $(C, D)$ respectively. Note that $\alpha \wedge \beta$ is represented by $(A \cap C, B \cup (A \cap D))$. For $f, g \in \mathcal{T}_{o}(X_{\bot})$, $\alpha[f, g] = \alpha[g, g]$ implies that $\alpha[f, g](x) = \alpha[g, g](x)$ for all $x \in X_{\bot}$. Consider

 \begin{equation*}
  (A, B)[f, g](x) = \begin{cases}
                     f(x), & \text{ if } x \in A; \\
                     g(x), & \text{ if } x \in B; \\
                     \bot, & \text{ otherwise.}
                    \end{cases}
 \end{equation*}

 Similarly we have

 \begin{equation*}
  (A, B)[g, g](x) = \begin{cases}
                     g(x), & \text{ if } x \in A; \\
                     g(x), & \text{ if } x \in B; \\
                     \bot, & \text{ otherwise.}
                    \end{cases}
 \end{equation*}

 Given that $(A, B)[f, g](x) = (A, B)[g, g](x)$ for all $x \in X_{\bot}$ as a consequence we have $f(x) = g(x)$ for all $x \in A$. Consider

 \begin{equation*}
  (A \cap C, B \cup (A \cap D))[f, g](x) = \begin{cases}
                                            f(x), & \text{ if } x \in A \cap C; \\
                                            g(x), & \text{ if } x \in B \cup (A \cap D); \\
                                            \bot, & \text{ otherwise.}
                                           \end{cases}
 \end{equation*}

 In view of the fact that $f(x) = g(x)$ for all $x \in A$ and that $A \cap C \subseteq A$ we have

 \begin{equation*}
  (A \cap C, B \cup (A \cap D))[f, g](x) = \begin{cases}
                                            g(x), & \text{ if } x \in A \cap C; \\
                                            g(x), & \text{ if } x \in B \cup (A \cap D); \\
                                            \bot, & \text{ otherwise.}
                                           \end{cases}
 \end{equation*}

 Thus $(\alpha \wedge \beta)[f, g] = (\alpha \wedge \beta)[g, g]$ and so quasi-identity \eqref{EC8} holds. \newline

\end{flushleft}

\subsection{Verification of \eref{ExampleBasicCset}} \label{VerBasicCset}

\begin{flushleft}

 Axiom \eqref{EC1}: By definition $U[s, t] = \bot$. \newline

 Axiom \eqref{EC6}: By definition $F[s, t] = t$. \newline

 Axiom \eqref{EC5}: We need to show that $(\neg \alpha)[s, t] = \alpha[t, s]$. It suffices to consider the following three cases:
 \begin{enumerate}[\rm]
  \item \emph{Case I}: $\alpha = T$: In this case $(\neg T)[s, t] = F[s, t] = t = T[t, s]$.
  \item \emph{Case II}: $\alpha = F$: In this case $(\neg F)[s, t] = T[s, t] = s = F[t, s]$.
  \item \emph{Case III}: $\alpha = U$: In this case $(\neg U)[s, t] = U[s, t] = \bot = U[t, s]$. \newline
  \end{enumerate}

 Axiom \eqref{EC3}: We need to verify the validity of the identity $\alpha[\alpha[s, t], u] = \alpha[s, u]$. It suffices to consider the following three cases:
 \begin{enumerate}[\rm]
  \item \emph{Case I}: $\alpha = T$: In this case $T[T[s, t], u] = T[s, t] = s = T[s, u]$.
  \item \emph{Case II}: $\alpha = F$: In this case $F[F[s, t], u] = u = F[s, u]$.
  \item \emph{Case III}: $\alpha = U$: In this case $U[U[s, t], u] = \bot = U[s, u]$. \newline
  \end{enumerate}

 Axiom \eqref{EC4}: We need to verify the validity of the identity $\alpha[s, \alpha[t, u]] = \alpha[s, u]$. It suffices to consider the following three cases:
 \begin{enumerate}[\rm]
  \item \emph{Case I}: $\alpha = T$: In this case $T[s, T[t, u]] = s = T[s, u]$.
  \item \emph{Case II}: $\alpha = F$: In this case $F[s, F[t, u]] = F[t, u] = u = F[s, u]$.
  \item \emph{Case III}: $\alpha = U$: In this case $U[s, U[t, u]] = \bot = U[s, u]$. \newline
  \end{enumerate}

 Axiom \eqref{EC7}: We need to verify the validity of the identity $(\alpha \wedge \beta)[s, t] = \alpha[\beta[s, t], t]$. It suffices to consider the following three cases:
 \begin{enumerate}[\rm]
  \item \emph{Case I}: $\alpha = T$: In this case $(T \wedge \beta)[s, t] = \beta[s, t] = T[\beta[s, t], t]$.
  \item \emph{Case II}: $\alpha = F$: In this case $(F \wedge \beta)[s, t] = F[s, t] = t = F[\beta[s, t], t]$.
  \item \emph{Case III}: $\alpha = U$: In this case $(U \wedge \beta)[s, t] = U[s, t] = \bot = U[\beta[s, t], t]$. \newline
  \end{enumerate}

 Axiom \eqref{EC2}: We need to verify the validity of the identity $\alpha[\beta[s, t], \beta[u, v]] = \beta[\alpha[s, u], \alpha[t, v]]$. It suffices to consider the following three cases:
 \begin{enumerate}[\rm]
  \item \emph{Case I}: $\alpha = T$: In this case $T[\beta[s, t], \beta[u, v]] = \beta[s, t] = \beta[T[s, u], T[t, v]]$.
  \item \emph{Case II}: $\alpha = F$: In this case $F[\beta[s, t], \beta[u, v]] = \beta[u, v] = \beta[F[s, u], F[t, v]]$.
  \item \emph{Case III}: $\alpha = U$: In this case $U[\beta[s, t], \beta[u, v]] = \bot$. Consider $\beta \in \{ T, F, U \}$. It is easy to see that in each case we have $\beta[\bot, \bot] = \bot$. In other words $U[\beta[s, t], \beta[u, v]] = \bot = \beta[\bot, \bot] = \beta[U[s, u], U[t, v]]$. \newline
  \end{enumerate}

 Axiom \eqref{EC8}: For verification of this quasi-identity we shall again consider the following three cases:
 \begin{enumerate}[\rm]
  \item \emph{Case I}: $\alpha = T$: The hypothesis $\alpha[s, t] = \alpha[t, t]$ gives $T[s, t] = T[t, t]$ that is $s = t$. It is easy to see that for each $\beta \in \{ T, F, U \}$ we have $\beta[s, t] = \beta[t, t]$ that is $(T \wedge \beta)[s, t] = (T \wedge \beta)[t, t]$.
  \item \emph{Case II}: $\alpha = F$: In this case $(F \wedge \beta)[s, t] = F[s, t] = t = F[t, t] = (F \wedge \beta)[t, t]$.
  \item \emph{Case III}: $\alpha = U$: In this case $(U \wedge \beta)[s, t] = U[s, t] = \bot = U[t, t] = (U \wedge \beta)[t, t]$. \newline
  \end{enumerate}

\end{flushleft}

\subsection{Verification of \eref{ExampleFunctionalAgreeable}} \label{VerFunctionalAgreeable}

\begin{flushleft}

Axiom \eqref{EA4}: We will show that $(f * f)[f, \zeta_{\bot}] = f$. For any $f \in \mathcal{T}_{o}(X_{\bot})$, $f * f = (A, B)$ where $A = \{ x \in X : f(x) \neq \bot \}$ and $B = \emptyset$. Thus
 \begin{align*}
 (A, B)[f, \zeta_{\bot}](x) & = \begin{cases}
                                f(x), & \text{ if } x \in A; \\
                                \bot, & \text{ if } x \notin A.
                               \end{cases} \\
                            & = \begin{cases}
                                f(x), & \text{ if } f(x) \neq \bot; \\
                                \bot, & \text{ if } f(x) = \bot.
                               \end{cases} \\
                            & = f(x)
 \end{align*} \newline

Axiom \eqref{EA1}: For any $f \in \mathcal{T}_{o}(X_{\bot})$, $\zeta_{\bot} * f = (A, B)$ where $A = \{ x \in X : \zeta_{\bot}(x) = f(x) \text{ } (\neq \bot) \}$ and $B = \{ x \in X : \zeta_{\bot}(x) \neq f(x) \text{ } (\neq \bot) \}$. Thus $\zeta_{\bot} * f = (A, B) = (\emptyset, \emptyset) = U = f * \zeta_{\bot}$. \newline

Axiom \eqref{EA2}: Let $f, g \in \mathcal{T}_{o}(X_{\bot})$ and $f * g = (A, B)$. Then
 \begin{equation*}
  (f * g)[f, g](x) =
  \begin{cases}
   f(x), & \text{ if } x \in A; \\
   g(x), & \text{ if } x \in B; \\
   \bot, & \text{ otherwise.}
  \end{cases}
 \end{equation*}
  Thus for $x \in A$ we have $f(x) = g(x) \text{ } (\neq \bot)$ and so $(f * g)[f, g](x) = f(x) = g(x) = (f * g)[g, g](x)$. Similarly when $x \in B$ we have $(f * g)[f, g](x) = g(x) = (f * g)[g, g](x)$. If $x \in (A \cup B)^{\mathsf{c}}$ then $(f * g)[f, g](x) = \bot = (f * g)[g, g](x)$. \newline

Axiom \eqref{EA3}: We will show that $\alpha[f, g] * \alpha[h, k] = \alpha \llbracket f * h, g * k \rrbracket$. Let $\alpha = (A, B)$ and $\mathscr{F}_{1} = \alpha[f, g]$, where
  \begin{equation*}
   \mathscr{F}_{1}(x) = \alpha[f, g](x) =
   \begin{cases}
    f(x), & \text{ if } x \in A; \\
    g(x), & \text{ if } x \in B; \\
    \bot, & \text{ otherwise.}
   \end{cases}
  \end{equation*}
  Similarly let $\mathscr{F}_{2} = \alpha[h, k]$ where
  \begin{equation*}
   \mathscr{F}_{2}(x) = \alpha[h, k](x) =
   \begin{cases}
    h(x), & \text{ if } x \in A; \\
    k(x), & \text{ if } x \in B; \\
    \bot, & \text{ otherwise.}
   \end{cases}
  \end{equation*}
  Let $\mathscr{F}_{1} * \mathscr{F}_{2} = (C, D)$ where $C = \{ x \in X : \mathscr{F}_{1}(x) = \mathscr{F}_{2}(x) (\neq \bot) \}$ and $D = \{ x \in X : \mathscr{F}_{1}(x) \neq \mathscr{F}_{2}(x) (\neq \bot) \}$. Let $f * h = (E, F)$ where $E = \{ x \in X : f(x) = h(x) (\neq \bot) \}$ and $F = \{ x \in X : f(x) \neq h(x) (\neq \bot) \}$. Let $g * k = (G, H)$ where $G = \{ x \in X : g(x) = k(x) (\neq \bot) \}$ and $H = \{ x \in X : g(x) \neq k(x) (\neq \bot) \}$. Then we have
  \begin{align*}
   (A, B) \llbracket (E, F), (G, H) \rrbracket & = \big( (A, B) \wedge (E, F) \big) \vee \big( \neg ((A, B)) \wedge (G, H) \big) \\
                                               & = \big( (A \cap E, B \cup (A \cap F) \big) \vee \big( B \cap G, A \cup (B \cap H) \big) \\
                                               & = \Big( (A \cap E) \cup \big( (B \cup (A \cap F)) \cap (B \cap G) \big), \big( B \cup (A \cap F) \big) \cap \big( A \cup (B \cap H) \big) \Big)
  \end{align*}
  In effect we need to show that $$(C, D) = \Big( (A \cap E) \cup \big( (B \cup (A \cap F)) \cap (B \cap G) \big), \big( B \cup (A \cap F) \big) \cap \big( A \cup (B \cap H) \big) \Big).$$

 \emph{Claim: $(A \cap E) \cup \big( (B \cup (A \cap F)) \cap (B \cap G) \big) \subseteq C$}. \\
   Let $x \in (A \cap E) \cup \big( (B \cup (A \cap F)) \cap (B \cap G) \big)$. Then $x \in (A \cap E)$ or $x \in \big( (B \cup (A \cap F)) \cap (B \cap G) \big)$.
   Suppose $x \in A \cap E$. It follows that $x \in A$ and so $\mathscr{F}_{1}(x) = f(x)$, $\mathscr{F}_{2}(x) = h(x)$. It is also true that $x \in E$ and so $f(x) = h(x) (\neq \bot)$. Thus $\mathscr{F}_{1}(x) = \mathscr{F}_{2}(x) (\neq \bot)$ from which it follows that $x \in C$.

   Suppose $x \in \big( (B \cup (A \cap F)) \cap (B \cap G) \big)$. Then $x \in B \cap G$ and hence $x \in B$, from which it follows that $\mathscr{F}_{1}(x) = g(x)$ and $\mathscr{F}_{2}(x) = k(x)$. Moreover as $x \in G$, $g(x) = k(x) (\neq \bot)$. Thus $\mathscr{F}_{1}(x) = \mathscr{F}_{2}(x) (\neq \bot)$ and so $x \in C$.

 \emph{Claim: $C \subseteq (A \cap E) \cup \big( (B \cup (A \cap F)) \cap (B \cap G) \big)$}. \\
   Let $x \in C$. Then $\mathscr{F}_{1}(x) = \mathscr{F}_{2}(x) (\neq \bot)$.
   \begin{enumerate}[\rm]
    \item \emph{Case I: $x \in A$:} From the fact that $\mathscr{F}_{1}(x) = \mathscr{F}_{2}(x) (\neq \bot)$, on assuming that $x \in A$, we have $\mathscr{F}_{1}(x) = f(x)$ and $\mathscr{F}_{2}(x) = h(x)$, that is, $f(x) = h(x) (\neq \bot)$, that is, $x \in E$. Hence $x \in (A \cap E)$ and so $x \in (A \cap E) \cup \big( (B \cup (A \cap F)) \cap (B \cap G) \big)$.
    \item \emph{Case II: $x \in B$:} If $x \in B$ then $\mathscr{F}_{1}(x) = g(x)$, $\mathscr{F}_{2}(x) = k(x)$ and as $\mathscr{F}_{1}(x) = \mathscr{F}_{2}(x) (\neq \bot)$, it follows that $g(x) = k(x) (\neq \bot)$. Hence $x \in B \cap G \subseteq B \subseteq (B \cup (A \cap F))$ from which the result follows.
    \item \emph{Case III: $x \in ( A \cup B)^{\mathsf{c}}$:} Then $\mathscr{F}_{1}(x) = \bot = \mathscr{F}_{2}(x)$. However, our assumption states that $\mathscr{F}_{1}(x) = \mathscr{F}_{2}(x) (\neq \bot)$, a contradiction. It follows that this case cannot occur under the given hypothesis.
   \end{enumerate}

 \emph{Claim: $\big( B \cup (A \cap F) \big) \cap \big( A \cup (B \cap H) \big) \subseteq D$}. \\
   Let $x \in \big( B \cup (A \cap F) \big) \cap \big( A \cup (B \cap H) \big)$. Then $x \in \big( B \cup (A \cap F) \big)$ and $x \in \big( A \cup (B \cap H) \big)$. From the fact that $x \in \big( B \cup (A \cap F) \big)$ we have $x \in B$ or $x \in (A \cap F)$.

   If $x \in B$ then from the fact that $A \cap B = \emptyset$, we have $x \notin A$. However, as it is given that $x \in \big( A \cup (B \cap H) \big)$, it must be the case that $x \in B \cap H$. Moreover, $x \in B$ implies that $\mathscr{F}_{1}(x) = g(x)$ and $\mathscr{F}_{2}(x) = k(x)$, and $x \in H$ implies that $g(x) \neq k(x) (\neq \bot)$. Consequently $\mathscr{F}_{1}(x) \neq \mathscr{F}_{2}(x) (\neq \bot)$, that is, $x \in D$.

   On the other hand, if $x \in (A \cap F)$, then $x \in A$ and so $\mathscr{F}_{1}(x) = f(x)$ and $\mathscr{F}_{2}(x) = h(x)$ and $x \in F$ means that $f(x) \neq h(x) (\neq \bot)$. Thus $\mathscr{F}_{1}(x) \neq \mathscr{F}_{2}(x) (\neq \bot)$, which means that $x \in D$.

 \emph{Claim: $D \subseteq \big( B \cup (A \cap F) \big) \cap \big( A \cup (B \cap H) \big)$}.
   Let $x \in D$. Then $\mathscr{F}_{1}(x) \neq \mathscr{F}_{2}(x) (\neq \bot)$.
   \begin{enumerate}[\rm]
    \item \emph{Case I: $x \in A$:} Then $\mathscr{F}_{1}(x) = f(x)$ and $\mathscr{F}_{2}(x) = h(x)$. From the hypothesis, we have $f(x) \neq h(x) (\neq \bot)$ which implies that $x \in F$. Hence $x \in (A \cap F) \subseteq A$, which gives that $x \in \big( B \cup (A \cap F) \big) \cap \big( A \cup (B \cap H) \big)$.
    \item \emph{Case II: $x \in B$:} In this case $\mathscr{F}_{1}(x) = g(x)$ and $\mathscr{F}_{2}(x) = k(x)$. From the hypothesis, it follows that $g(x) \neq k(x) (\neq \bot)$ which implies that $x \in H$. Thus $x \in B \cap H \subseteq B$ which implies that $x \in \big( B \cup (A \cap F) \big) \cap \big( A \cup (B \cap H) \big)$.
    \item \emph{Case III: $x \in (A \cup B)^{\mathsf{c}}$:} Then $\mathscr{F}_{1}(x) = \bot = \mathscr{F}_{2}(x)$, a contradiction to our assumption. It follows that this case cannot occur.
   \end{enumerate}
  Thus $(C, D) = \Big( (A \cap E) \cup \big( (B \cup (A \cap F)) \cap (B \cap G) \big), \big( B \cup (A \cap F) \big) \cap \big( A \cup (B \cap H) \big) \Big)$. \newline

Axiom \eqref{EA5}: Let $f * f = (X, \emptyset)$ and $f * g = (\emptyset, \emptyset)$ for $f, g \in \mathcal{T}_{o}(X_{\bot})$. It follows from $f * f = (X, \emptyset)$ that $\{ x \in X : f(x) \neq \bot \} = X$ that is, for all $x \in X$, $f(x) \neq \bot$. Consider $f * g = (A, B) = (\emptyset, \emptyset)$ where $A = \{ x \in X : f(x) = g(x) (\neq \bot) \}$ and $B = \{ x \in X : f(x) \neq g(x) (\neq \bot) \}$. Suppose that for some $x \in X$, we have $g(x) \neq \bot$. Using the fact that $f(x) \neq \bot$ and that $A = \emptyset$, we can deduce that $f(x) \neq g(x)$, and so by definition, $x \in B$. However, it is given that $B = \emptyset$, a contradiction. Hence, $g(x) = \bot$ for each $x \in X$, that is $g = \zeta_{\bot}$. Therefore the last quasi-identity holds.
\end{flushleft}

\subsection{Verification of \eref{ExampleBasicAgreeable}} \label{VerBasicAgreeable}

\begin{flushleft}

Axiom \eqref{EA4}: For $s \in S_{\bot} \setminus \{ \bot \}$, $s * s = T$. Thus $(s * s)[s, \bot] = T[s, \bot] = s$. If $s = \bot$ then $s * s = \bot * \bot = U$. Thus $(s * s)[s, \bot] = U[\bot, \bot] = \bot = s$. \newline

Axiom \eqref{EA1}: $\bot * s = U = s * \bot$ by definition. \newline

Axiom \eqref{EA2}: For $s, t \in S_{\bot}$, we have the following:
  \begin{enumerate}[\rm]
   \item \emph{Case I: $s * t = T$:} This occurs if and only if $s = t$ and $s \neq \bot \neq t$. Thus $(s * t)[s, t] = T[s, t] = s = t = T[t, t] = (s * t)[t, t]$.
   \item \emph{Case II: $s * t = F$:} In this case, we have $(s * t)[s, t] = F[s, t] = t = F[t, t] = (s * t)[t, t]$.
   \item \emph{Case III: $s * t = U$:} Here we have $(s * t)[s, t] = U[s, t] = \bot = U[t, t] = (s * t)[t, t]$. \newline
  \end{enumerate}

Axiom \eqref{EA3}: It suffices to consider the following three cases:
  \begin{enumerate}[\rm]
   \item \emph{Case I: $\alpha = T$:} Then $\alpha[s, t] * \alpha[u, v] = T[s, t] * T[u, v] = s * u$. And $\alpha \llbracket s * u, t * v \rrbracket = (T \wedge (s * u)) \vee (F \wedge (t * v)) = s * u$.
   \item \emph{Case II: $\alpha = F$:} Then $\alpha[s, t] * \alpha[u, v] = F[s, t] * F[u, v] = t * v$. And $\alpha \llbracket s * u, t * v \rrbracket = (F \wedge (s * u)) \vee (T \wedge (t * v)) = t * v$.
   \item \emph{Case III: $\alpha = U$:} Then $\alpha[s, t] * \alpha[u, v] = U[s, t] * U[u, v] = \bot * \bot = U$. And $\alpha \llbracket s * u, t * v \rrbracket = (U \wedge (s * u)) \vee (U \wedge (t * v)) = U \vee U = U$. \newline
  \end{enumerate}

Axiom \eqref{EA5}: Let $s, t \in S_{\bot}$ such that $s * s = T$ and $s * t = U$. Clearly $s \neq t$. As $s * s = T$, this means that $s \neq \bot$. If $t \neq \bot$ then in view of the fact that $s \neq t$ we have $s * t = F$ by definition, which is a contradiction to our assumption that $s * t = U$. Thus $t = \bot$.
\end{flushleft}

\subsection{Verification of \eref{ExampleCAlgAgreeable}} \label{VerCAlgAgreeable}

\begin{flushleft}

 Axiom \eqref{EA4}: For $\alpha \in M$, $\alpha * \alpha = (\alpha \wedge \alpha) \vee (\neg \alpha \wedge \neg \alpha) = \alpha \vee \neg \alpha$. Hence $(\alpha * \alpha) \llbracket \alpha, U \rrbracket = (\alpha \vee \neg \alpha) \llbracket \alpha, U \rrbracket = ((\alpha \vee \neg \alpha) \wedge \alpha) \vee (\neg \alpha \wedge \alpha \wedge U)$. This reduces to checking the validity of the identity $((\alpha \vee \neg \alpha) \wedge \alpha) \vee (\neg \alpha \wedge \alpha \wedge U) = \alpha$ in the $C$-algebra $M$. In view of \tref{SubdirIrredCAlg}, it suffices to check over three elements $T, F, U$.
 \begin{enumerate}[\rm]
  \item \emph{Case I}: $\alpha = T$: $((\alpha \vee \neg \alpha) \wedge \alpha) \vee (\neg \alpha \wedge \alpha \wedge U) = ((T \vee F) \wedge T) \vee (F \wedge T \wedge U) = T \vee F = T = \alpha$.
  \item \emph{Case II}: $\alpha = F$: $((\alpha \vee \neg \alpha) \wedge \alpha) \vee (\neg \alpha \wedge \alpha \wedge U) = ((F \vee T) \wedge F) \vee (T \wedge F \wedge U) = F \vee F = F = \alpha$.
  \item \emph{Case III}: $\alpha = U$: $((\alpha \vee \neg \alpha) \wedge \alpha) \vee (\neg \alpha \wedge \alpha \wedge U) = ((U \vee U) \wedge U) \vee (U \wedge U \wedge U) = U \vee U = U = \alpha$. \newline
 \end{enumerate}

 Axiom \eqref{EA1}: For $\alpha \in M$, $U * \alpha = (U \wedge \alpha) \vee (\neg U \wedge \neg \alpha) = U \vee U = U$. On the other hand $\alpha * U = (\alpha \wedge U) \vee (\neg \alpha \wedge \neg U)$. It suffices to consider the following three cases:
 \begin{enumerate}[\rm]
  \item \emph{Case I}: $\alpha = T$: $\alpha * U = T * U = (T \wedge U) \vee (F \wedge \neg U) = U \vee F = U$.
  \item \emph{Case II}: $\alpha = F$: $\alpha * U = F * U = (F \wedge U) \vee (T \wedge \neg U) = F \vee U = U$.
  \item \emph{Case III}: $\alpha = U$: $\alpha * U = U * U = (U \wedge U) \vee (U \wedge \neg U) = U \vee U = U$. \newline
 \end{enumerate}

 Axiom \eqref{EA2}: For $\alpha, \beta \in M$ we have $(\alpha * \beta) \llbracket \alpha, \beta \rrbracket = ((\alpha \wedge \beta) \vee (\neg \alpha \wedge \neg \beta)) \llbracket \alpha, \beta \rrbracket = (((\alpha \wedge \beta) \vee (\neg \alpha \wedge \neg \beta)) \wedge \alpha) \vee (\neg ((\alpha \wedge \beta) \vee (\neg \alpha \wedge \neg \beta)) \wedge \beta)$. On the other hand $(\alpha * \beta) \llbracket \beta, \beta \rrbracket = (((\alpha \wedge \beta) \vee (\neg \alpha \wedge \neg \beta)) \wedge \beta) \vee ( \neg ((\alpha \wedge \beta) \vee (\neg \alpha \wedge \neg \beta)) \wedge \beta)$. It suffices to check the validity of this identity in the following three cases:
 \begin{enumerate}[\rm]
  \item \emph{Case I}: $\alpha = T$: Here $(\alpha * \beta) \llbracket \alpha, \beta \rrbracket = (((T \wedge \beta) \vee (F \wedge \neg \beta)) \wedge T) \vee (\neg ((T \wedge \beta) \vee (F \wedge \neg \beta)) \wedge \beta) = ((\beta \vee F) \wedge T) \vee (\neg (\beta \vee F) \wedge \beta) = \beta \vee (\neg \beta \wedge \beta) = \beta$. Also $(\alpha * \beta) \llbracket \beta, \beta \rrbracket = (((T \wedge \beta) \vee (F \wedge \neg \beta)) \wedge \beta) \vee ( \neg ((T \wedge \beta) \vee (F \wedge \neg \beta)) \wedge \beta) = ((\beta \vee F) \wedge \beta) \vee (\neg (\beta \vee F) \wedge \beta) = \beta \vee (\neg \beta \wedge \beta) = \beta$.
  \item \emph{Case II}: $\alpha = F$: In this case $(\alpha * \beta) \llbracket \alpha, \beta \rrbracket = (((F \wedge \beta) \vee (T \wedge \neg \beta)) \wedge F) \vee (\neg ((F \wedge \beta) \vee (T \wedge \neg \beta)) \wedge \beta) = ((F \vee \neg \beta) \wedge F) \vee (\neg (F \vee \neg \beta) \wedge \beta) = (\neg \beta \wedge F) \vee (\beta \wedge \beta) = (\neg \beta \wedge \beta) \vee \beta$ using the fact that $x \wedge F = x \wedge \neg x$ holds in all $C$-algebras. This expression reduces to $(\neg \beta \wedge \beta) \vee \beta = \beta$. Also $(\alpha * \beta) \llbracket \beta, \beta \rrbracket = (((F \wedge \beta) \vee (T \wedge \neg \beta)) \wedge \beta) \vee ( \neg ((F \wedge \beta) \vee (T \wedge \neg \beta)) \wedge \beta) = ((F \vee \neg \beta) \wedge \beta) \vee (\neg (F \vee \neg \beta) \wedge \beta) = (\neg \beta \wedge \beta) \vee (\beta \wedge \beta) = (\neg \beta \wedge \beta) \vee \beta = \beta$.
  \item \emph{Case III}: $\alpha = U$: In this case $(\alpha * \beta) \llbracket \alpha, \beta \rrbracket = (((U \wedge \beta) \vee (U \wedge \neg \beta)) \wedge U) \vee (\neg ((U \wedge \beta) \vee (U \wedge \neg \beta)) \wedge \beta) = ((U \vee U) \wedge U) \vee (\neg (U \vee U) \wedge \beta) = U$. Also $(\alpha * \beta) \llbracket \beta, \beta \rrbracket = (((U \wedge \beta) \vee (U \wedge \neg \beta)) \wedge \beta) \vee ( \neg ((U \wedge \beta) \vee (U \wedge \neg \beta)) \wedge \beta) = ((U \vee U) \wedge \beta) \vee (\neg (U \vee U) \wedge \beta) = U$. \newline
 \end{enumerate}

 Axiom \eqref{EA3}: We need to check the validity of the identity $\alpha \llbracket \beta, \gamma \rrbracket * \alpha \llbracket \delta, \rho \rrbracket = \alpha \llbracket \beta * \delta, \gamma * \rho \rrbracket$. It suffices to consider the following three cases:
 \begin{enumerate}[\rm]
  \item \emph{Case I}: $\alpha = T$: In this case $\alpha \llbracket \beta, \gamma \rrbracket * \alpha \llbracket \delta, \rho \rrbracket = T \llbracket \beta, \gamma \rrbracket * T \llbracket \delta, \rho \rrbracket = \beta * \delta$. On the other hand $\alpha \llbracket \beta * \delta, \gamma * \rho \rrbracket = T \llbracket \beta * \delta, \gamma * \rho \rrbracket = \beta * \delta$.
  \item \emph{Case II}: $\alpha = F$: Here $\alpha \llbracket \beta, \gamma \rrbracket * \alpha \llbracket \delta, \rho \rrbracket = F \llbracket \beta, \gamma \rrbracket * F \llbracket \delta, \rho \rrbracket = \gamma * \rho$. Also $\alpha \llbracket \beta * \delta, \gamma * \rho \rrbracket = F \llbracket \beta * \delta, \gamma * \rho \rrbracket = \gamma * \rho$.
  \item \emph{Case III}: $\alpha = U$: We have $\alpha \llbracket \beta, \gamma \rrbracket * \alpha \llbracket \delta, \rho \rrbracket = U \llbracket \beta, \gamma \rrbracket * U \llbracket \delta, \rho \rrbracket = U$. Also $\alpha \llbracket \beta * \delta, \gamma * \rho \rrbracket = U \llbracket \beta * \delta, \gamma * \rho \rrbracket = U$. \newline
 \end{enumerate}

 Axiom \eqref{EA5}: In order to verify quasi-identity \eqref{EA5}, we recall from \tref{SubdirIrredCAlg} that $M$ is a subalgebra of $\mathbb{3}^{X}$ for some set $X$. Suppose that $\alpha * \alpha = \mathbf{T}$ and $\alpha * \beta = \mathbf{U}$ for some $\alpha, \beta \in M$, that is, $\alpha \vee \neg \alpha = \mathbf{T}$ and $(\alpha \wedge \beta) \vee (\neg \alpha \wedge \neg \beta) = \mathbf{U}$. Treating $\alpha, \beta$ as elements of $\mathbb{3}^{X}$, we have $\alpha(x) \vee \neg (\alpha(x)) = T$, and $(\alpha(x) \wedge \beta(x)) \vee (\neg \alpha(x) \wedge \neg \beta(x)) = U$ for each $x \in X$. Since $\alpha(x) \vee \neg \alpha(x) = T$, where $\alpha(x) \in \{ T, F, U \}$, there are only two possible cases:
 \begin{enumerate}
  \item \emph{Case I}: $\alpha(x) = T$: Since $(\alpha * \beta)(x) = U$ we have $(\alpha(x) \wedge \beta(x)) \vee (\neg \alpha(x) \wedge \neg \beta(x)) = (T \wedge \beta(x)) \vee (F \wedge \neg \beta(x)) = \beta(x) = U$.
  \item \emph{Case II}: $\alpha(x) = F$: Since $(\alpha * \beta)(x) = U$ we have $(\alpha(x) \wedge \beta(x)) \vee (\neg \alpha(x) \wedge \neg \beta(x)) = (F \wedge \beta(x)) \vee (T \wedge \neg \beta(x)) = \neg \beta(x) = U$, from which it follows that $\beta(x) = U$.
 \end{enumerate}
 Thus for each $x \in X$, $\beta(x) = U$ that is, $\beta = \mathbf{U}$. Therefore, the quasi-identity \eqref{EA5} holds.
\end{flushleft}

\end{document}